\documentclass[journal, draftclsnofootnote]{IEEEtran}  
\usepackage{amssymb}
\usepackage{mathtools}
\usepackage{amsmath}
\usepackage{graphicx}
\usepackage{float}
\usepackage{epsfig}
\usepackage{csquotes}
\usepackage[font=footnotesize,labelfont=footnotesize]{caption}
\usepackage{subcaption}
\usepackage{cleveref}
\usepackage[T1]{fontenc}
\usepackage{color}
\usepackage{cite}
\usepackage{lettrine}
\usepackage{url}
\usepackage{amsthm}
\usepackage{enumitem}
\usepackage[table]{xcolor} 
\usepackage{array}
\usepackage{easybmat}
\usepackage{multirow,bigdelim}
\usepackage{diagbox}
\usepackage{float}
\usepackage{algorithm}
\usepackage[noend]{algpseudocode}
\usepackage{nicefrac}
\DeclareMathOperator*{\argmin}{\arg\min}

\usepackage{setspace}
\usepackage{tensor}
\algrenewcommand{\algorithmiccomment}[1]{\vskip0.1ex$\triangleright$ #1}

\usepackage{tikz}
\usepackage[utf8]{inputenc}
\usepackage{pgfplots} 
\usepackage{pgfgantt}
\usepackage{pdflscape}
\pgfplotsset{compat=newest} 
\pgfplotsset{plot coordinates/math parser=false}

\newlength\fwidth
\newlength\fheight

\usepackage[usestackEOL]{stackengine}
\stackMath

 \newtheorem{theorem}{Theorem}

 \newtheorem{corollary}[]{Corollary}
\newcommand\Tstrut{\rule{0pt}{3ex}}  

\usepackage[keeplastbox,nospread]{flushend}

\begin{document}
\newcommand{\Matr}{\begin{bmatrix}}
\newcommand{\matr}{\end{bmatrix}}

\newcommand{\Acal}{{\cal A}}
\newcommand{\Bcal}{{\cal B}}
\newcommand{\Dcal}{{\cal D}}
\newcommand{\Ical}{{\cal I}}
\newcommand{\Kcal}{{\cal K}}
\newcommand{\Rn}{{\rm I \! R}}
\newcommand{\Cn}{{\rm C \hspace*{-1ex} \rule[0.1ex]{0.15ex}{1.3ex} \hspace{1ex}}}
\newcommand{\Cnsmall}{{\rm C \hspace*{-0.75ex} \rule[0.1ex]{0.1ex}{0.9ex} \hspace{0.75ex}}}
\newcommand{\Pcal}{{\cal P}}
\newcommand{\Rcal}{{\cal R}}
\newcommand{\Qcal}{{\cal Q}}
\newcommand{\Wcal}{{\cal W}}
\newcommand{\Xcal}{{\cal X}}
\newcommand{\Zcal}{{\cal Z}}

\newcommand{\trace}{{\rm tr} \,}
\newcommand{\diag}{{\bf Diag} \,}
\newcommand{\diagB}{{\bf diag} \,}
\newcommand{\ev}{{\bf ev} \,}
\newcommand{\sv}{{\bf sv} \,}
\newcommand{\Real}{{\rm Re} \,}
\newcommand{\Imag}{{\rm Im} \,}
\newcommand{\Even}{{\rm Ev} \,}
\newcommand{\Odd}{{\rm Od} \,}
\newcommand{\sgn}{{\rm sgn}}

\newcommand{\aB}{{\bf a}}
\newcommand{\bB}{{\bf b}}
\newcommand{\bBa}{\underline{\bf b}}
\newcommand{\cB}{{\bf c}}
\newcommand{\eB}{{\bf e}}
\newcommand{\eBa}{\underline{\bf e}}
\newcommand{\fB}{{\bf f}}
\newcommand{\fBa}{\underline{\bf f}}
\newcommand{\gB}{{\bf g}}
\newcommand{\gBa}{\underline{\bf g}}
\newcommand{\hB}{{\bf h}}
\newcommand{\jB}{{\bf j}}
\newcommand{\kB}{{\bf k}}
\newcommand{\mB}{{\bf m}}
\newcommand{\mBa}{\underline{\bf m}}
\newcommand{\nB}{{\bf n}}
\newcommand{\nBa}{\underline{\bf n}}
\newcommand{\pB}{{\bf p}}
\newcommand{\qB}{{\bf q}}
\newcommand{\ra}{{\underline{r}}}
\newcommand{\rB}{{\bf r}}
\newcommand{\rBc}{{\widetilde{\bf r}}}
\newcommand{\rBa}{\underline{\bf r}}
\newcommand{\sB}{{\bf s}}
\newcommand{\sBc}{{\widetilde{\bf s}}}
\newcommand{\sBa}{\underline{\bf s}}
\newcommand{\tB}{{\bf t}}
\newcommand{\uB}{{\bf u}}
\newcommand{\vB}{{\bf v}}
\newcommand{\vBa}{\underline{\bf v}}
\newcommand{\wB}{{\bf w}}
\newcommand{\wBa}{\underline{\bf w}}
\newcommand{\xB}{{\bf x}}
\newcommand{\xBa}{\underline{\bf x}}
\newcommand{\ya}{\underline{y}}
\newcommand{\yB}{{\bf y}}
\newcommand{\yBa}{\underline{\bf y}}
\newcommand{\zB}{{\bf z}}

\newcommand{\epsilonB}{{\bm \epsilon}}
\newcommand{\PiB}{{\bm \Pi}}
\newcommand{\gammaB}{{\bm \gamma}}
\newcommand{\SigmaB}{{\bm \Sigma}}
\newcommand{\sigmaB}{{\bm \sigma}}
\newcommand{\sigmaBc}{\widetilde{\bm \sigma}}
\newcommand{\sigmaBa}{\underline{\bm \sigma}}
\newcommand{\rhoc}{\widetilde{\rho}}
\newcommand{\rhoB}{{\bm \rho}}
\newcommand{\rhoBc}{\widetilde{\bm \rho}}
\newcommand{\rhoBa}{\underline{\bm \rho}}
\newcommand{\PhiBa}{\underline{\bm \Phi}}
\newcommand{\psiB}{{\bm \psi}}
\newcommand{\PsiB}{{\bm \Psi}}
\newcommand{\PsiBa}{\underline{\bm \Psi}}
\newcommand{\tauB}{{\bm \tau}}
\newcommand{\LambdaB}{{\bf \Lambda}}
\newcommand{\LambdaBa}{\underline{\bm \Lambda}}
\newcommand{\lambdaB}{{\bm \lambda}}
\newcommand{\muB}{{\bm \mu}}
\newcommand{\muBa}{\underline{\bm \mu}}
\newcommand{\XiB}{{\bm \Xi}}
\newcommand{\xiB}{{\bm \xi}}
\newcommand{\xiBc}{\widetilde{\bm \xi}}
\newcommand{\xiBa}{\underline{\bm \xi}}
\newcommand{\xiBo}{\overline{\bm \xi}}
\newcommand{\omegaB}{{\bm \omega}}
\newcommand{\omegaBc}{\widetilde{\bm \omega}}
\newcommand{\omegaBa}{\underline{\bm \omega}}
\newcommand{\omegaBo}{\overline{\bm \omega}}
\newcommand{\UpsilonB}{{\bm \Upsilon}}
\newcommand{\thetaB}{{\bm \theta}}
\newcommand{\thetaBa}{\underline{\bm \theta}}
\newcommand{\ThetaB}{{\bm \Theta}}

\newcommand{\AB}{{\bf A}}
\newcommand{\ABa}{\underline{\bf A}}
\newcommand{\ABb}{\Bbb{A}}
\newcommand{\ABc}{\widetilde{\bf A}}
\newcommand{\ABbc}{\widetilde{\Bbb{A}}}
\newcommand{\ABba}{\underline{\Bbb{A}}}
\newcommand{\BB}{{\bf B}}
\newcommand{\BBa}{\underline{\bf B}}
\newcommand{\BBc}{\widetilde{\bf B}}
\newcommand{\CB}{{\bf C}}
\newcommand{\CBa}{\underline{\bf C}}
\newcommand{\CBc}{\widetilde{\bf C}}
\newcommand{\DB}{{\bf D}}
\newcommand{\DBb}{\Bbb{D}}
\newcommand{\EB}{{\bf E}}
\newcommand{\FB}{{\bf F}}
\newcommand{\FBa}{\underline{\bf F}}
\newcommand{\FBc}{\widetilde{\bf F}}
\newcommand{\GB}{{\bf G}}
\newcommand{\GBa}{\underline{\bf G}}
\newcommand{\GBc}{\widetilde{\bf G}}
\newcommand{\HB}{{\bf H}}
\newcommand{\JB}{{\bf J}}
\newcommand{\JBc}{\tilde{\bf J}}
\newcommand{\JBba}{\underline{\Bbb{J}}}
\newcommand{\HBc}{\widetilde{\bf H}}
\newcommand{\HBa}{\underline{\bf H}}
\newcommand{\HBb}{\Bbb{H}}
\newcommand{\HBbc}{\widetilde{\Bbb{H}}}
\newcommand{\HBba}{\underline{\Bbb{H}}}
\newcommand{\IB}{{\bf I}}
\newcommand{\JBa}{\underline{\bf J}}
\newcommand{\KB}{{\bf K}}
\newcommand{\KBb}{\Bbb{K}}
\newcommand{\KBa}{\underline{\bf K}}
\newcommand{\KBo}{\overline{\bf K}}
\newcommand{\KBc}{\widetilde{\bf K}}
\newcommand{\LB}{{\bf L}}
\newcommand{\LBa}{\underline{\bf L}}
\newcommand{\MB}{{\bf M}}
\newcommand{\MBa}{\underline{\bf M}}
\newcommand{\NB}{{\bf N}}
\newcommand{\NBa}{\underline{\bf N}}
\newcommand{\PB}{{\bf P}}
\newcommand{\PBa}{\underline{\bf P}}
\newcommand{\PBb}{\Bbb{P}}
\newcommand{\PBbc}{\widetilde{\Bbb{P}}}
\newcommand{\PBba}{\underline{\Bbb{P}}}
\newcommand{\PBc}{\widetilde{\bf P}}
\newcommand{\QB}{{\bf Q}}
\newcommand{\QBc}{\widetilde{\bf Q}}
\newcommand{\QBa}{\underline{\bf Q}}
\newcommand{\QBb}{\Bbb{Q}}
\newcommand{\QBbc}{\widetilde{\Bbb{Q}}}
\newcommand{\QBba}{\underline{\Bbb{Q}}}
\newcommand{\RB}{{\bf R}}
\newcommand{\RBa}{\underline{\bf R}}
\newcommand{\RBb}{\Bbb{R}}
\newcommand{\RBba}{\underline{\Bbb{R}}}
\newcommand{\RBc}{\widetilde{\bf R}}
\newcommand{\SB}{{\bf S}}
\newcommand{\SBc}{\widetilde{\bf S}}
\newcommand{\SBa}{\underline{\bf S}}
\newcommand{\SBb}{\Bbb{S}}
\newcommand{\SBba}{\underline{\Bbb{S}}}
\newcommand{\TB}{{\bf T}}
\newcommand{\TBa}{\underline{\bf T}}
\newcommand{\TBc}{\widetilde{\bf T}}
\newcommand{\UB}{{\bf U}}
\newcommand{\UBa}{\underline{\bf U}}
\newcommand{\VB}{{\bf V}}
\newcommand{\VBa}{\underline{\bf V}}
\newcommand{\WB}{{\bf W}}
\newcommand{\WBc}{\widetilde{\bf W}}
\newcommand{\WBa}{\underline{\bf W}}
\newcommand{\XB}{{\bf X}}
\newcommand{\XBa}{\underline{\bf X}}
\newcommand{\YB}{{\bf Y}}
\newcommand{\YBa}{\underline{\bf Y}}
\newcommand{\ZB}{{\bf Z}}
\newcommand{\BZ}{{\bf 0}}
\newcommand{\oneB}{{\bf 1}}

\title{ Determining the Dimension and Structure of the  Subspace Correlated Across Multiple Data Sets} 
\author{\IEEEauthorblockN{Tanuj Hasija,~\IEEEmembership{Student Member,~IEEE,} Christian Lameiro,~\IEEEmembership{Member,~IEEE,}  Timothy Marrinan,~\IEEEmembership{Member,~IEEE,} and  Peter J. Schreier,~\IEEEmembership{Senior~Member,~IEEE}\\}
\thanks{ The authors are with the Signal and System Theory Group, Universit\"{a}t Paderborn 33098, Germany (e-mail: tanuj.hasija@sst.upb.de, christian.lameiro@sst.upb.de, timothy.marrinan@sst.upb.de, peter.schreier@sst.upb.de).}
\thanks{This research was supported by the German Research Foundation (DFG) under grant SCHR 1384/3-2.}} 

\maketitle

\begin{abstract}
\textbf{ Detecting the components common or correlated across multiple data sets is challenging due to a large number of possible correlation structures among the components. Even more challenging is to determine the precise structure of these correlations. Traditional work has focused on determining only the model order, i.e., the dimension of the correlated subspace, a number that depends on how the model-order problem is defined. Moreover, identifying the model order is often not enough to understand the relationship among the components in different data sets. We aim at solving the complete model-selection problem, i.e., determining which components are correlated across which data sets. We prove that the eigenvalues and eigenvectors of the normalized covariance matrix of the composite data vector, under certain conditions, completely characterize the underlying correlation structure.  We use these results to solve the model-selection problem by employing bootstrap-based hypothesis testing. }
\end{abstract}


\section{Introduction}

Canonical correlation analysis (CCA),  multiset CCA (mCCA), and their variants have been widely used for jointly analyzing the relationships among different sets of data \cite{hotelling1936relations,kettenring1971canonical,asendorf2015improving,carroll1968generalization,li2009joint}. These tools extract components from each data set, referred to as canonical variables,  that are highly correlated across different data sets. The canonical variables can then be used in a plethora of applications including data fusion in biomedicine \cite{correa2010canonical,adali2015multimodal,levin2016sample}, feature extraction in machine learning \cite{hardoon2004canonical}, finding coupled patterns in oceanography \cite{tippett2008regression}, extracting gene clusters in genomics \cite{yamanishi2003extraction}, etc. Despite the abundance of interest in extracting correlated components, one key question often remains overlooked. Is the estimated correlation of these components actually present in the underlying system, or is it an artifact due to noise or an insufficient number of samples?
Sometimes the answer to this question is assumed to be known a priori from domain-specific knowledge. However, outside of that limited realm, most applications employ methods for thresholding the correlation coefficients as a way to determine the significant correlated components. These solutions are generally heuristic and often fail for one of two main reasons: 1) when the number of samples is limited, the correlation coefficient among the estimated components is overestimated, even to the point of identifying nonexistent, spurious correlations, 2) the number of possible correlation structures among the extracted components combinatorially increases as the number of data sets increases.   

Some techniques in the past have aimed to solve this as a \textit{model-order selection} problem. In signal processing, \enquote{model order} is the term used for the dimension of a parameter vector, i.e., the number of parameters of the data model \cite{stoica2004model}. Thus, estimating the number of correlated components can be posed as a model-order selection problem. For two data sets, the signal components are either correlated or uncorrelated across both sets, and the model-order selection problem is well-defined. In this case, counting the number of correlated components also answers the question of which components are correlated. Some of the techniques for estimating the model order for two data sets are \cite{stoica1996nonparametric,chen1996detection,Song2016canonical}. 	However, for more than two data sets, the model-order selection problem is ambiguous. It is possible for the components to be correlated across no data sets, all data sets, or some subset of the collection. In this context, many generalizations of model-order selection are valid, and the problem must be more precisely defined. Fig. \ref{fig:Example_Introduction} illustrates an example of correlation structure between the latent signal components of three data sets, $\sB_1$, $\sB_2$, and $\sB_3$. In most of cases, however, we observe linear mixtures of these latent components instead of observing them directly. Nonetheless, Fig. \ref{fig:Example_Introduction} can be used as a reference example to differentiate between the possible definitions of model-order selection for multiple data sets.  Each column of the figure indicates the components of one data set and thus components of each column have the same subscript. Each row represents the individual signal components that can be correlated between different pairs of data sets. In this example each data set contains five signal components that are mutually uncorrelated within their set. The set $\mathbf{s}_1$ is repeated again in the last column to illustrate the correlation between the first and the third data set. Here, the first signal component  of each data set is correlated with the first component of each other data set (the correlation is indicated with red arrows), and the next three components  are correlated only between a pair of data sets (indicated with black arrows). The fifth component of each data set is uncorrelated with all other data sets.

One variation of the model-order selection problem for multiple data sets is to determine the number of signal components correlated across every pair of data sets. In Fig. \ref{fig:Example_Introduction}, only the first signal components demonstrate this property, so the model order by this definition would be one. Model orders of this type can be identified using the methods described in \cite{wu2002determination,hasija2016detecting,songdetermining,hasija2016bootstrap,bhinge2017estimation}. The methods in \cite{wu2002determination} and \cite{hasija2016detecting} require the assumption that if signal components are correlated across any data sets then they must be correlated across all data sets, and these methods fail to estimate the model order when this assumption is not met. This assumption is relaxed in \cite{songdetermining} and generalized to arbitrary correlation structures in \cite{hasija2016bootstrap} and \cite{bhinge2017estimation}.

\begin{figure}[t]
\centering
\includegraphics[scale=1]{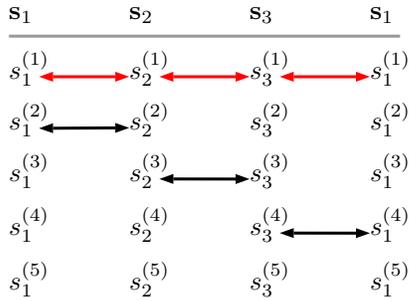}
\caption{  Example of a correlation structure between the latent signal components of three data sets, $\mathbf{s}_1$, $\mathbf{s}_2$ and $\mathbf{s}_3$. The first component is correlated across all the data sets, the next three are correlated across two data sets only and the fifth component in each data set is uncorrelated with the other components.}
 \label{fig:Example_Introduction}
\end{figure}

Alternatively, signal components correlated across a subset of the collection of data sets might also be of interest. For instance, when tracking an object in videos recorded by spatially separated cameras, the object might not be visible in every frame of each camera \cite{asendorf2015improving}. Thus if multiset canonical correlation is used to measure the similarity of frames from the different views, it would be pertinent for the model order to represent the number of components correlated across all data sets or a subset of data sets.  Similarly in brain imaging, estimating the number of signal components activated in the fMRI data of all the subjects is useful for multi-subject analysis \cite{correa2010canonical,calhoun2012multisubject}. However, some brain regions may not appear active for some subjects, due to noise or other factors, even if biological intuition suggests that they should be. Knowing that correlations exist among the signal components corresponding to these regions is useful even if they are only present in a subset of subjects.
These scenarios suggest that an appropriate definition of model order should count the signals that demonstrate correlation across all or a subset of data sets. By this definition the model order of the example in Fig. \ref{fig:Example_Introduction} is four. In \cite{asendorf2015improving}, the authors propose a test statistic to estimate this model order and showed  when this number can be correctly estimated for different settings of the signal-to-noise ratio (SNR) and the number of samples. 

In the end, determining only the model order is insufficient to completely characterize the correlation structure in multiple data sets. This summary statistic only provides the knowledge that the components exhibit correlation. This knowledge, although sufficient for two data sets, is incomplete for multiple data sets as it is also required to determine which components are correlated across which data sets. For the example in Fig. \ref{fig:Example_Introduction}, the complete solution is not just determining that the first four components are correlated but also that the first component in each data set  is correlated, and the successive components are correlated between data sets $1$ and $2$, $2$ and $3$, and $1$ and $3$, respectively. Therefore, we formulate and solve a more general \textit{model-selection} problem, which includes model-order selection as a subproblem. In this work, we propose a technique for solving this model-selection problem for multiple data sets using the eigenvalues and eigenvectors of the coherence matrix (normalized covariance matrix \cite{scharf2000canonical}) of the composite data set. The main contributions of this work are: 
\begin{itemize}
\item We prove that, under fairly general conditions, the correlation structure in multiple data sets can be fully characterized from the eigenvector decomposition of the coherence matrix of the composite data.
\item Using this theoretical result, we develop an algorithm that identifies the correlation structure effectively in practical scenarios. 
\end{itemize}

A competing solution to the model-selection problem was recently presented by \cite{marrinancomplete}.  The technique in \cite{marrinancomplete} applies a series of hypothesis tests to pairs of components extracted using CCA. It first determines the number of pairwise correlations, two data sets at a time, and amalgamates these detections using joint information from the complete collection of data sets. The technique proposed in this work has both practical and theoretical benefits over \cite{marrinancomplete} in the right circumstances. First, this work provides conditions under which the correlation structure can be identified, and justifies these conditions theoretically. Second, the proposed method relies solely on joint information from all of the data sets, and as a result, it demonstrates superior detection accuracy for components correlated across more data sets. Finally, for $P$ number of data sets, the proposed solution requires significantly fewer hypothesis tests per component; running only $P$ tests compared to the $2^P$ needed for \cite{marrinancomplete}.

Our program for this paper is as follows. After defining the problem in Section \ref{sec:prob_form}, we derive the main result in Section \ref{sec:coherence_matrix_proves}. In Section \ref{sec:bt_ht}, we present two hypothesis testing procedures  to determine the underlying correlation structure based on the results derived in Section \ref{sec:coherence_matrix_proves}. We use the bootstrap to estimate the distribution of the test statistic under the null hypothesis. Finally, in Section \ref{sec:results} simulations show that our technique reliably estimates the complete correlation structure of multiple data sets, and is competitive with the existing state-of-the-art approaches. 
 
\section{Problem Formulation}
\label{sec:prob_form}

We consider $P$ data sets consisting of zero-mean, real-valued random vectors, $\mathbf{x}_1,\ldots,\mathbf{x}_P$, each of dimension $n$. The data sets are generated by an unknown linear mixing of  underlying signal vectors, $\mathbf{s}_1,\ldots,\mathbf{s}_P \in \mathbb{R}^n$. The generating data model is \begin{equation}
\label{eq:DataModel}
\mathbf{x}_p = \mathbf{A}_p \mathbf{s}_p,  \quad p=1,2,\ldots,P,
\end{equation}
where $\mathbf{A}_p \in \mathbb{R}^{{n}\times n}$  is an unknown but fixed mixing matrix with full rank\footnote{All the results in this work can be easily extended to the general case where the data sets have different dimensions and the mixing matrices are non-square with full column rank. We omit the general case as the extension is trivial when the inverses are replaced by pseudo-inverses, because the more cumbersome bookkeeping distracts from the actual result.  }. 
We make the following additional assumptions: 
\begin{itemize}
\item  The signal vectors each contain $n$ signal components. The $i$th signal component of the $p$th data set is denoted by $s_{p}^{(i)}$.
These components are assumed to be zero-mean, unit variance, and uncorrelated within each data set, i.e.,
\begin{align}
E[\mathbf{s}_p] &= \mathbf{0} \qquad \text{and}\\
\mathbf{R}_{s_p s_p} &= E[\mathbf{s}_p\mathbf{s}_p^T] = \mathbf{I},
\end{align}
where $\mathbf{0}$ is a zero-vector, and $\mathbf{I}$ is an identity matrix. 
\item Only signal components with the same component number can have nonzero correlation. As such, the signal cross-covariance matrix between data sets $p$ and $q$ ($p \neq q$) is 
\begin{equation}
\RB_{s_p s_q} = \text{diag}(\rho_{pq}^{(1)},\rho_{pq}^{(2)}, \ldots, \rho_{pq}^{(n)} ),
\end{equation}
where diag($.$) denotes a diagonal matrix with specified diagonal elements, and $\rho_{pq}^{(i)} $ represents the unknown (possibly zero) correlation coefficient between their $i$th components. When analyzing the correlations between two data sets only, $\RB_{s_p s_q}$ can be assumed to be diagonal without loss of generality. For more than two sets, diagonal cross-covariance matrices are a restriction on the problem, as they do not represent all possible correlation structures. However, this assumption makes the multiset problem tractable and is commonplace in the literature \cite{li2009joint,anderson2014independent}. 
\end{itemize}

Our aim is to completely determine the underlying correlation structure among the signal components $\sB_p$ based on $M$ independent and identically distributed (i.i.d.) joint samples of the data vectors $\xB_p$, $p=1,\ldots,P$. 
We solve this problem by first estimating the model order $d$, which represents the total number of components that demonstrate nonzero correlation i.e., $d = | \{ i | \exists p,q \text{ for which } \rho_{pq}^{(i)} \neq 0\} |$. 
 Then we estimate the data sets across which these $d$ components are correlated, i.e., we identify all data sets $p$, $q$ for which $\rho_{pq}^{(i)} \neq 0$. 

To compare to existing techniques, two scalar model orders are also defined. We denote the number of signal components correlated pairwise between data sets $p$ and $q$ as $d_{pq}$. The number of signals correlated across all data sets is denoted by $d_{\textrm{all}}$, i.e., $d_{\textrm{all}} = | \{ i | \rho_{pq}^{(i)} \neq 0\text{ }\forall p,q |$. Prior works have either focused on determining $d_{pq}$ for two data sets \cite{Song2016canonical} or on estimating only $d_\text{all}$ for more than two data sets \cite{wu2002determination,hasija2016detecting,songdetermining,hasija2016bootstrap,bhinge2017estimation}. However, as discussed earlier, even knowing both $d_{pq}$ and $d_\text{all}$ is not enough to completely determine the underlying correlation structure  (except with very special types of correlation structures). 

\textbf{Examples}: Let us revisit the example in Fig. \ref{fig:Example_Introduction}. Table \ref{tab:Example_1} provides one example of correlation coefficients that match the structure presented in Fig. \ref{fig:Example_Introduction}.  The entries are the correlation coefficients between signal components of different pairs of data sets. In this case $d_{pq} = 2$ for all choices of $p$ and $q$, $d_\text{all} = 1$, and $d = 4$. Now consider an example of $4$ data sets each with $4$ signal components as shown in Table \ref{tab:Example_2}. The first component of data sets $2,3,$ and $4$ and the second component of data sets $1,3,$ and $4$ are correlated. The third and fourth components of data sets $1$ and $2$ are correlated. Hence, $d_\text{all} = 0$, $d = 4$ and $d_{pq}$ is the number of nonzero entries in the corresponding column. In both these examples, the techniques in \cite{Song2016canonical,wu2002determination,hasija2016detecting,songdetermining,hasija2016bootstrap,bhinge2017estimation} provide solutions for either $d_{pq}$ or $d_\text{all}$. Our technique, however, aims to identify which of the entries in Tables  \ref{tab:Example_1} and \ref{tab:Example_2} are nonzero. 
\begin{table}
\begin{center}
 \begin{tabular}{ c | c c c } 
 & $\rho_{12}^{(i)}$ & $\rho_{13}^{(i)}$ & $\rho_{23}^{(i)}$  \\[1ex]
   \hline \Tstrut
 $i=1$ & 0.5 & 0.6 & 0.6  \\[.5ex]
 $i=2$ & 0.7 & 0 & 0 \\[.5ex]
 $i=3$ & 0 & 0 & 0.8  \\[.5ex]
  $i=4$ & 0 & 0.4 & 0  \\[.5ex]
 $i=5$ & 0 & 0 & 0 
\end{tabular}
\end{center}
\caption{ Example with the correlation structure in Fig. \ref{fig:Example_Introduction} with three data sets each with five signal components. The entries are the correlation coefficients between signal components of different pairs of data sets.  } 
\label{tab:Example_1}
\end{table}

\begin{table}
\begin{center}
 \begin{tabular}{ c | c c c c c c} 
 & $\rho_{12}^{(i)}$ & $\rho_{13}^{(i)}$ & $\rho_{14}^{(i)}$ & $\rho_{23}^{(i)}$ & $\rho_{24}^{(i)}$ & $\rho_{34}^{(i)}$ \\[1ex]
  \hline \Tstrut
 $i=1$ & 0 & 0 & 0 & 0.7 & 0.2 & 0.8 \\[.5ex]
 $i=2$ & 0 & 0.6 & 0.4 & 0 & 0 & 0.5 \\[.5ex]
 $i=3$ & 0.5 & 0 & 0 & 0 & 0 & 0 \\[.5ex]
 $i=4$ & 0.5 & 0 & 0 & 0 & 0 & 0 
\end{tabular}
\end{center}
\caption{ Example of $4$ data sets each with $4$ signal components.   } 
\label{tab:Example_2}
\end{table}

\section{Correlated Subspace in Multiple Data Sets}
\label{sec:coherence_matrix_proves}

The problem as stated in Section \ref{sec:prob_form} requires joint knowledge of the relationships between all data sets.  In this section, our main results demonstrate that the pertinent information can be found in the eigenvector decomposition of the whitened covariance matrix of the concatenation of all data sets.  Consider the composite data vector $\mathbf{x}$ obtained by vertically concatenating the individual data vectors, 
\begin{equation}
\mathbf{x} = [\mathbf{x}_1^T, \ldots, \mathbf{x}_P^T]^T, 
\end{equation}
and the composite covariance matrix $\mathbf{R} = E[\mathbf{x} \mathbf{x}^T]$. Similarly, the composite signal vector is defined as $\sB = [\mathbf{s}_1^T, \ldots, \mathbf{s}_P^T]^T$, and the composite signal covariance matrix as $\RB_{ss} = E[\sB \sB^T]$.  The definition of the coherence matrix for two data sets \cite{scharf2000canonical} can be generalized in a natural way for this composite data as
\begin{equation}
\label{eq:aug_coh_matrix}
\mathbf{C} = \mathbf{R}_D^{-\frac{1}{2}} \mathbf{R} \mathbf{R}_D^{-\frac{1}{2}}.
\end{equation}
Here $\RB_D = \text{blkdiag}(\mathbf{R}_{11}, \ldots, \mathbf{R}_{PP})$ is a block-diagonal matrix with $\RB_{pp} = E[\xB_p \xB_p^T]$, and exponent \({- \frac{1}{2}}\) on $\RB_D$ denotes the inverse of the matrix square root. The composite coherence matrix, $\mathbf{C}$, can be written in a block structure given as
\begin{equation}
\label{eq:Cblock}
\CB =
 \begin{bmatrix}
  \mathbf{I} & \mathbf{C}_{12} & \cdots & \mathbf{C}_{1P} \\
  \mathbf{C}_{21} & \mathbf{I} & \cdots & \mathbf{C}_{2P} \\
  \vdots  & \vdots  & \ddots & \vdots  \\
  \mathbf{C}_{P1} & \mathbf{C}_{P2} & \cdots & \mathbf{I} 
 \end{bmatrix} ,
\end{equation}
where the diagonal blocks are identity matrices and the off-diagonal blocks are the coherence matrices of two data sets. For each pair of data sets $\xB_p$ and $\xB_q$, the coherence matrix can be decomposed as 
\begin{align}
\CB_{pq}  &= \mathbf{R}_{pp}^{-\frac{1}{2}} \mathbf{R}_{pq} \mathbf{R}_{qq}^{-\frac{1}{2}},\notag \\
		&= (\AB_p \RB_{s_p s_p} \AB_p^T)^{-\frac{1}{2}}  \AB_p \RB_{s_p s_q} \AB_q^T (\AB_q \RB_{s_q s_q} \AB_q^T)^{-\frac{1}{2}},  \notag \\
		&= (\AB_p  \AB_p^T)^{-\frac{1}{2}}  \AB_p \RB_{s_p s_q} \AB_q^T (\AB_q \AB_q^T)^{-\frac{1}{2}},
\end{align}
since $\RB_{s_p s_p} = \RB_{s_q s_q} =\IB$. Let $\FB_p = ({\AB}_p {\AB}_p^T )^{-1/2} {\AB}_p$ and similarly ${\FB}_q = ({\AB}_q {\AB}_q^T )^{-1/2} {\AB}_q$ so that  ${\FB}_p {\FB}_p^T = {\FB}_q {\FB}_q^T = {\IB} $, and we have 
\begin{align}
\label{eq:svd_cij}
\mathbf{C}_{pq} &= \mathbf{F}_p \RB_{s_p s_q} \mathbf{F}_q^T . 
\end{align}
 
 Using \eqref{eq:Cblock} and \eqref{eq:svd_cij}, the composite coherence matrix $\CB$ can be written as
 \begin{equation}
 \label{eq:Cxx_dec2}
\CB = \FB \RB_{ss} \FB^T,
 \end{equation}
 where $\FB = \text{blkdiag}(\FB_1, \ldots, \FB_{P})$. Based on the assumption that  $\RB_{s_p s_q}$ is diagonal, the elements of $\RB_{ss}$ can be permuted to form a block-diagonal matrix whose $i$th block is the covariance matrix formed by the $i$th components of each data set. That is, there exists a permutation $\PB$ where
\begin{align}
\label{eq:Cxx_perm}
\CB &= \FB \PB^T \PB \RB_{ss} \PB^T\PB \FB^T,  \notag\\ 
		&=  \FB \PB^T \widetilde{\RB}_{ss} \PB \FB^T, 
\end{align}
such that $\widetilde{\RB}_{ss} = \PB \RB_{ss} \PB^T$ is a block-diagonal matrix defined as 
\begin{equation}
\label{eq:Rsst}
\widetilde{\RB}_{ss} = \text{blkdiag}(\RB^{(1)},\ldots,\RB^{(n)}),
\end{equation} 
and $\RB^{(i)} \in \mathbb{R}^{P \times P}$ is the covariance matrix of the $i$th components of each data set. 

\subsection{Eigenvalues of $\CB$} 
\label{sec:Eval_Cxx}

We now illustrate an explicit relationship between the dimension of the correlated subspace, $d$, and the eigenvalues of $\CB$. It can be observed from \eqref{eq:Cblock} that $\CB$ is an identity matrix when all signal components are uncorrelated, and thus all eigenvalues of $\CB$ are one. However, when some components are correlated, $\CB $ has eigenvalues that are different from one. More specifically, when data sets demonstrate correlations across $d$ components, $\CB$ has at least $d$ eigenvalues greater than one. A key question then is: \textit{when is the dimension of the correlated subspace, $d$, exactly equal to the number of eigenvalues of $\CB$ greater than one?} 
 
 The answer is not as straightforward as one would hope. However, we can identify a set of sufficient conditions under which this property holds. The property also often holds when these conditions are not met and the algorithm proposed in this work is robust to the violation of these assumptions. The proof of the sufficiency of these conditions relies on three things: 
\begin{enumerate}[label=\roman*)]
\item \label{1cond} The composite coherence matrix $\CB$ is {similar} (through a similarity transformation) to a block diagonal matrix where each block contains a non-identity matrix corresponding to a collection of data sets whose $i$th components are all correlated with each other and an identity matrix corresponding to data sets whose $i$th components are uncorrelated. 
\item \label{2cond} If the $i$th components of four or more data sets are correlated with each other, all nonzero correlations are greater than a prescribed threshold.
\item \label{3cond} All correlations are transitive. That is, if a signal component is correlated between data sets $p$ and $q$, and between data sets $q$ and $r$, then it is also correlated between data sets $p$ and $r$.  
\end{enumerate}
For simplicity, we prove the result in Theorem \ref{theorem:eval} when there is only one block of the form described in \ref{1cond} for the $i$th components of all data sets, but of course the result holds for any number of such blocks. We will discuss the requirement of \ref{2cond} and \ref{3cond} during the proof of Theorem \ref{theorem:eval}.

\begin{theorem}
\label{theorem:eval}
 Let $\CB$ be the composite coherence matrix of $P$ data sets constructed according to the linear mixing model in \eqref{eq:DataModel} with pairwise diagonal signal cross-covariance matrices. Let $k^{(i)}$ be the number of data sets whose $i$th components are correlated. Assume that correlations are transitive, and for $k^{(i)} \geq 4$, each correlation coefficient is either $\rho_{pq}^{(i)} = 0$ or $\rho_{pq}^{(i)} > \epsilon^{(i)} = (\frac{k^{(i)}-1}{k^{(i)}})^2$ for all $p,q$. Let $\mathcal{I} = \{1,\ldots,n\}$ be an index set for the signals.
 $\CB$ has exactly $d$ eigenvalues greater than one if and only if there exists a subset of signals $\mathcal{D} \subseteq \mathcal{I}$ with $|\mathcal{D}|= d$, and for each $i \in \mathcal{D}$ there exists a $p \neq q$ such that  $s_p^{(i)}$ and $s_q^{(i)}$ are correlated. 
\end{theorem}

\begin{proof}

We begin by showing that if there exists a subset of signals $\mathcal{D} \subseteq \mathcal{I}$ with $|\mathcal{D}|= d,$ and for each $i \in \mathcal{D}$  there exists a $p,q$  with $s_p^{(i)}$ and $s_q^{(i)}$ correlated, $\CB$ has $d$ eigenvalues greater than one. Let  $\widetilde{\RB}_{ss}$ be an $nP \times nP$ matrix with the structure defined in \eqref{eq:Rsst}. The diagonal blocks of $\widetilde{\RB}_{ss}$ can be indexed by the associated signal component.  That is,   $\RB^{(i)} \in \mathbb{R}^{P \times P}$ is the covariance matrix of the $i$th components of each data set with the form
\begin{equation}
\label{eq:Ri}
\RB^{(i)} = \begin{bmatrix}
\BB^{(i)} & \mathbf{0} \\
\mathbf{0} & \IB
\end{bmatrix}.
\end{equation}
 $\BB^{(i)} \in \mathbb{R}^{k^{(i)} \times k^{(i)}}$ is a symmetric matrix with diagonal elements equal to one and off-diagonal elements equal to the correlation coefficients between the correlated $i$th components.  The indices associated with the subset of data sets whose $i$th components are correlated is $\mathcal{K}^{(i)} \subseteq \mathcal{P} = \{1, \ldots P\}$, and the dimensions of $\BB^{(i)}$ are the size of this subset, $|\mathcal{K}^{(i)}| = k^{(i)}$.  As $\widetilde{\RB}_{ss}$ is block-diagonal, its eigenvalues are equal to the eigenvalues of the blocks $\RB^{(i)}$. 
Since $\FB \PB^T$ is an orthogonal matrix, the eigenvalues of $\CB$ are equal to the eigenvalues of $\widetilde{\RB}_{ss}$ and therefore, equal to the eigenvalues of $\RB^{(i)}$. 

Let  $\BB^{(i)}$ be decomposed as $\BB^{(i)} = \IB + \HB^{(i)}$ for each $i$. $\HB^{(i)}$ is a hollow (with zeros on the diagonal) symmetric matrix whose off-diagonal elements are the nonzero correlation coefficients corresponding to the $i$th component.  We show that the matrix $\HB^{(i)}$ has exactly one positive eigenvalue for each $i$  as follows:
\\ \underline{Case 1 ($k^{(i)}=2$)}:  $\HB^{(i)}$ has exactly one positive eigenvalue for all values of $\rho_{pq}^{(i)} >0$, because 
\begin{equation}
\HB^{(i)} = \begin{bmatrix}
0 & \rho_{pq}^{(i)} \\
\rho_{pq}^{(i)} & 0
\end{bmatrix},
\end{equation}
which has eigenvalues $\{ \rho_{pq}^{(i)} ,-\rho_{pq}^{(i)} \}$.
\\ \underline{Case 2 ($k^{(i)}=3$)}: When the $i$th components of exactly three data sets, $p,q$ and $r$, are correlated, $\HB^{(i)}$ is  
\begin{equation}
\HB^{(i)} = \begin{bmatrix}
0 & \rho_{pq}^{(i)}  & \rho_{pr}^{(i)}\\
\rho_{pq}^{(i)} & 0 &  \rho_{qr}^{(i)} \\
\rho_{pr}^{(i)} & \rho_{qr}^{(i)} &  0
\end{bmatrix}.
\end{equation}
The characteristic polynomial of $\HB^{(i)}$ is $y = -\lambda^3 + \lambda(\rho_{pq}^{(i)2}+\rho_{pr}^{(i)2}+\rho_{qr}^{(i)2} ) + 2\rho_{pq}^{(i)}\rho_{pr}^{(i)}\rho_{qr}^{(i)}$. Using Descartes' rule of sign change, $y$ has only one positive root for any $\rho_{pq}^{(i)},\rho_{pr}^{(i)},\rho_{qr}^{(i)} > 0$ \cite{struik1969source}. Therefore, $\HB^{(i)}$ has only one positive eigenvalue.  
\\ \underline{Case 3 ($k^{(i)} \geq 4$)}: \cite[Theorem 3.5]{charles2013nonpositive} can be used to show that $\HB^{(i)}$ has exactly one positive eigenvalue if all the off-diagonal elements of $\HB^{(i)}$ are greater than $\epsilon^{(i)} = (\frac{k^{(i)}-1}{k^{(i)}})^2$. This result is demonstrated in Appendix \ref{app:lem1}. Without the assumption of transitive correlations, $\HB^{(i)}$ cannot be guaranteed to have all the positive off-diagonal elements as required. 
Thus, $\HB^{(i)}$ has exactly one positive eigenvalue in each case as desired. 

Let  $\mathcal{D} \subseteq \mathcal{I}$ be the $d$ values of $i$ for which correlation exists. For each $i \in \mathcal{D}$,  $\HB^{(i)}$ has one positive eigenvalue and $k^{(i)}-1$ non-positive eigenvalues. Let $\HB^{(i)} = \UB \mathbf{\Lambda}^{(i)} \UB^T$ be the eigenvalue decomposition (EVD) of $\HB^{(i)}$ with $\UB \UB^T = \IB$. The EVD of $\BB^{(i)}$ can be written as $\BB^{(i)} = \UB(\IB + \mathbf{\Lambda}^{(i)}) \UB^T$ so that, for each $i \in \mathcal{D}$, $\BB^{(i)}$ has one eigenvalue greater than one and the remaining $k^{(i)}-1$ eigenvalues less than or equal to one. Using \eqref{eq:Ri}, the maximum eigenvalue of $\BB^{(i)}$ is also the maximum eigenvalue of $\RB^{(i)}$, implying that $\RB^{(i)}$ has exactly one eigenvalue greater than one. Hence $\CB$ has exactly $d$ eigenvalues greater than one as desired. 

We now show that the converse is also true. Let $\CB$ has $d$ eigenvalues greater than one. There exists a permutation as described by \eqref{eq:Cxx_perm} where
\begin{equation}
\RB^{(i)} = \begin{bmatrix}
\BB^{(i)} & \mathbf{0} \\
\mathbf{0} & \IB
\end{bmatrix},
\end{equation}
and $\BB^{(i)}$ is the sum of an identity matrix and a hollow symmetric matrix of correlation coefficients, i.e., $\BB^{(i)} = \IB + \HB^{(i)}$, as described above.  When no data sets are correlated for a given $i$, $\RB^{(i)} = \IB$. Let $\mathcal{D}' \subseteq \mathcal{I}$ be the indices for which nonzero correlation exists (and by assumption is greater than $\epsilon^{(i)}$ for $k^{(i)} \geq 4$).  Thus, 
\begin{equation}
\RB^{(i)} = \begin{bmatrix}
\IB + \HB^{(i)} & \mathbf{0} \\
\mathbf{0} & \IB
\end{bmatrix},
\end{equation}
 for $i \in \mathcal{D}'$ and $\RB^{(i)} = \IB$ for $i \in \mathcal{I}\textbackslash\mathcal{D}'$.  Clearly $\RB^{(i)}$ has no eigenvalues greater than one for $i \in \mathcal{I}\textbackslash\mathcal{D}'$.

Suppose that $|\mathcal{D}'|<d$.  Then there exists an $i$ for which $\HB^{(i)}$ has more than one eigenvalue greater than zero. However, this contradicts our proof that $\HB^{(i)}$ has exactly one positive eigenvalue for each $i$.  Similarly, suppose that $|\mathcal{D}'|>d$. Then there exists an $i$ for which $\HB^{(i)}$ has no eigenvalues greater than zero.  This also contradicts our proof, and thus $|\mathcal{D}'|=d$ and there are exactly $d$ values of $i$ for which $\rho_{pq}^{(i)} \neq 0$ for some $p \neq q$.
\end{proof}

Theorem \ref{theorem:eval} guarantees that, when  $d$ eigenvalues of $\CB$ are greater than one, $d$ signal components are correlated across at least a pair of data sets. Therefore, the number of correlated signal components can be determined by testing for the number of eigenvalues of $\CB$ that are greater than one. Such a test is formulated in Section \ref{sec:Eval_test} using bootstrap-based hypothesis testing. 

 One of the assumptions in Theorem \ref{theorem:eval} that needs further discussion  is that if the $i$th components of four or more data sets are correlated (i.e., $k^{(i)} \geq 4$), then  the correlation coefficient between any pair of $i$th signal components  must be either zero or greater than $\epsilon^{(i)} = (\frac{k^{(i)}-1}{k^{(i)}})^2$. This assumption guarantees that only one eigenvalue of $\CB$ corresponding to the $i$th component is greater than one. The threshold, $\epsilon^{(i)}$, is derived in Appendix \ref{app:lem1} and is a restrictive threshold since $\lim_{k^{(i)} \rightarrow \infty} \epsilon^{(i)} = 1$. However, the proof does not claim to represent all matrices with the desired eigenvalue structure. That is, there is a nonempty set of real positive hollow symmetric matrices that have exactly one positive eigenvalue but do not meet this element-wise threshold. One example is the following: Suppose the nonzero correlation coefficients associated with the $i$th component are equal for all $k^{(i)} (\geq 4)$ data sets.  That is,  
$\rho_{pq}^{(i)} = \rho^{(i)}$ $\forall p,q \in \mathcal{K}^{(i)}$, where $\mathcal{K}^{(i)}$ is the subset of indices associated with the data sets whose $i$th components are correlated.  In this case, 
\begin{equation}
\label{eq:B_homog}
\HB^{(i)} = \begin{bmatrix}
0 & \rho^{(i)}  & \cdots  & \rho^{(i)} \\
\rho^{(i)} & 0 &  & \vdots  \\
\vdots & & \ddots  & \vdots \\
\rho^{(i)} & \cdots & \cdots &  0
\end{bmatrix},
\end{equation}
as defined in the proof of Theorem \ref{theorem:eval}. This can be simplified as $\HB^{(i)} = \rho^{(i)} \mathbf{1}\mathbf{1}^T - \rho^{(i)}\IB $, where $\mathbf{1} \in \mathbb{R}^{k^{(i)}}$ is a vector with all elements equal to one. The maximum eigenvalue of $\HB^{(i)}$ is $(k^{(i)}-1)\rho^{(i)}$ and the remaining $k^{(i)}-1$ eigenvalues are $-\rho^{(i)}$. Therefore, $\HB^{(i)}$ has one positive eigenvalue for any $\rho^{(i)} > 0$.   In this example, the relationship between the eigenvalues of $\CB$ and the number of signals with nonzero correlations described by Theorem \ref{theorem:eval} holds true for $0 < \rho_{pq}^{(i)} \leq 1$. 

In the general case, even though $\epsilon^{(i)}$ is restrictive, an element-wise threshold like this is perhaps the best that can be hoped for without imposing further constraints on the structure of the correlation among the components.  As noted in \cite{charles2013nonpositive}, for any $k \in \mathbb{N}$ with $k \geq 3$, there exists a positive hollow symmetric $\HB \in \mathbb{R}^{k \times k}$ such that $\HB$ has only two nonpositive eigenvalues. That is to say, without the element-wise constraint there will always be feasible correlation structures for which $\CB$ has more eigenvalues greater than one than signals with nonzero correlations. Moreover, as we will see in our numerical examples later, our hypothesis-test based techniques presented in Section \ref{sec:bt_ht} may still perform satisfactorily even in cases where the assumptions of Theorem \ref{theorem:eval} are violated.

As an immediate consequence of Theorem \ref{theorem:eval}, any eigenvalue of $\CB$ that is equal to the maximum possible value, $P$, identifies a signal component where all $P$ data sets are perfectly correlated.  This is shown in Corollary \ref{cor1}.
\begin{corollary}
\label{cor1}
If any eigenvalue of $\CB$, $\lambda^{(i)}$, is equal to $P$, there exists an $i \in \mathcal{D}$ such that the correlation between $s_p^{(i)}$ and $s_q^{(i)}$ is one for all $p,q=1,\ldots,P$. 
\end{corollary}
\begin{proof}
By Theorem \ref{theorem:eval}, if $\lambda^{(i)} > 1$  is an eigenvalue of $\CB$, there exists an $\RB^{(i)} \neq \IB$ whose largest eigenvalue is equal to $\lambda^{(i)}$. The diagonal elements of $\mathbf{R}^{(i)}$ are equal to one by definition, so $\text{trace}(\RB^{(i)}) = P = \sum_{j=1}^{P} \lambda_j^{(i)}$,
where $\lambda_j^{(i)}$ is the $j$th largest eigenvalue of $\RB^{(i)}$.  Since the largest eigenvalue of $\RB^{(i)}$, $\lambda_1^{(i)} = \lambda^{(i)} = P$, we have $\lambda^{(i)}_2 = \ldots = \lambda^{(i)}_P = 0,$ and thus, the rank of $\RB^{(i)}$ is one. 

Let $\RB^{(i)} = \wB \wB^T$ be a rank-one matrix, where $\wB = [w_1, \ldots, w_P]^T \in \mathbb{R}^P$. The diagonal elements of $\RB^{(i)}$ are equal to one, implying that $w_p^2 = 1$ for all $p$. Since the off-diagonal elements are bounded by zero and one, $w_p$ is positive for $p=1,\ldots,P$. Thus, $w_1 = w_2 = \ldots = w_P = 1$ is the only solution for $\wB,$ and $\RB^{(i)} = \mathbf{1} \mathbf{1}^T$. Therefore, $\rho^{(i)}_{pq} = 1 \forall p,q$ and the $i$th component of each data set is perfectly correlated with the $i$th component of all other data sets. 
\end{proof}

\subsection{Eigenvectors of $\CB$}
\label{sec:Evec of Cxx}

The eigenvalues of $\CB$ provide information about the dimension of the correlated subspace, but identifying exactly which data sets demonstrate correlation in a particular component requires more information than this summary contains. The eigenvectors of $\CB$, on the other hand, contain as their elements the coefficients for constructing the correlated signals from each data set that correspond to the associated eigenvalue (that is greater than one). Data sets connected to the nonzero elements of an eigenvector are then the ones whose components are correlated among the associated group of components. 

\begin{theorem}
\label{theorem:evec}
Let $\CB$ be the composite coherence matrix of $P$ data sets constructed according to the linear mixing model in \eqref{eq:DataModel} with pairwise diagonal signal cross-covariance matrices. Let $k^{(i)}$ be the number of data sets whose $i$th components are correlated. Assume that correlations are transitive, and for $k^{(i)} \geq 4$, each correlation coefficient is either $\rho_{pq}^{(i)} = 0$ or $\rho_{pq}^{(i)} > \epsilon^{(i)} = (\frac{k^{(i)}-1}{k^{(i)}})^2$ for all $p,q$.
Let $\CB \uB^{(i)} = \lambda^{(i)} \uB^{(i)}$ such that $\lambda^{(i)} > 1$ is an eigenvalue with algebraic multiplicity of one, and let  
the eigenvector $\uB^{(i)}$ be partitioned into $P$ subvectors, $\uB^{(i)} = [\uB_1^{(i)T}, \uB_2^{(i)T}, ..., \uB_P^{(i)T}]^T,$ where $\uB_p^{(i)} \in \mathbb{R}^{n}$ contains the elements of $\uB^{(i)}$ associated with the $p$th data set.  Then the $i$th signal component in the $p$th data set is among the group of correlated $i$th components if and only if $\uB_p^{(i)} \neq \mathbf{0}$. 
\end{theorem}

\begin{proof}
Let the $i$th components of $k^{(i)}$ data sets be correlated, and  let $\mathcal{K}^{(i)} \subseteq \mathcal{P} = \{1,\ldots,P \}$ be the subset of indices associated with these correlated data sets so that $|\mathcal{K}^{(i)}|=k^{(i)}$. Suppose the $i$th signal component in the $p$th data set is among the correlated components, i.e., $p$ is an element of $\mathcal{K}^{(i)}$. 
 Then there exists a permutation as described by \eqref{eq:Cxx_perm} where $\RB^{(i)} \in \mathbb{R}^{P \times P} \neq \IB$ such that  
\begin{equation}
\label{eq:Rievec}
\RB^{(i)} = \begin{bmatrix}
\BB^{(i)} & \mathbf{0} \\
\mathbf{0} & \IB
\end{bmatrix}, 
\end{equation}
where $\BB^{(i)} \in \mathbb{R}^{k^{(i)} \times k^{(i)}}$ is a symmetric and element-wise positive matrix that contains the correlation coefficients between the correlated $i$th components.  According to Theorem 1, the largest eigenvalue of $\BB^{(i)}$ is greater than one and satisfies $\BB^{(i)} \vB^{(i)} = \lambda_{\text{max}}^{(i)} \vB^{(i)},$ where $\vB^{(i)}$ is the eigenvector of $\BB^{(i)}$ associated with $\lambda_{\text{max}}^{(i)}$. According to the Perron-Frobenius theorem, all the entries of $\vB^{(i)}$ are positive \cite{bellman1997introduction}.
Since $\RB^{(i)}$ is block-diagonal with blocks $\BB^{(i)}$ and $\IB$, $\lambda_{\text{max}}^{(i)}$ is also the largest eigenvalue of $\RB^{(i)}$ with the eigenvector 
\begin{equation}
\tilde{\vB}^{(i)} = \begin{bmatrix} 
{\vB}^{(i)} \\ \mathbf{0} \end{bmatrix},
\end{equation}
where $\mathbf{0}$ is a zero vector of dimensions $P-k^{(i)}$.   

$\RB^{(i)}$ is a block of $\widetilde{\RB}_{ss}$ as defined in \eqref{eq:Rsst}, and thus $\lambda_{\max}^{(i)}$ is also an eigenvalue of  $\widetilde{\RB}_{ss}$. Each eigenvector, $\tilde{\uB}^{(i)}$, of $\widetilde{\RB}_{ss}$ can be partitioned into $n$ subvectors where the elements of the $i$th subvector correspond to the $i$th block of $\widetilde{\RB}_{ss}$.  Since $\lambda_{\text{max}}^{(i)}$ has an algebraic multiplicity of one, the eigenvector of $\widetilde{\RB}_{ss}$ associated with $\lambda_{\text{max}}^{(i)}$ has $\widetilde{\vB}^{(i)}$ at the $i$th position and zeros everywhere else and  can be written as $\tilde{\uB}^{(i)} = [\mathbf{0}, \ldots, \tilde{\vB}^{(i)T}, \ldots, \mathbf{0}]^T$.  
Using \eqref{eq:Cxx_perm}, the eigenvector of $\CB$ associated with $\lambda_{\text{max}}^{(i)}$ is related to $\tilde{\uB}^{(i)}$ by 
\begin{equation}
\label{eq:eveccxx}
\uB^{(i)} = \FB \PB^T \tilde{\uB}^{(i)}.  
  \end{equation}
For each $p\in\mathcal{K}^{(i)}$, the element of $\tilde{\vB}^{(i)}$ corresponding to the $p$th data set is strictly positive. Therefore, from the definition of $\tilde{\uB}^{(i)}$ and \eqref{eq:eveccxx},  ${\mathbf{u}}^{(i)}_p \neq \mathbf{0}$ as desired. 

To show the other implication, suppose that the $i$th signal component in the $p$th data set is uncorrelated among the $i$th group of components. Thus, $p$ is not an element of $\mathcal{K}^{(i)}$. Therefore, the $p$th data set is not represented in $\vB^{(i)}$ but rather corresponds to one of the elements in the zero vector of $\tilde{\vB}^{(i)}$. Using the definition of $\tilde{\uB}^{(i)}$  and \eqref{eq:eveccxx}, it is easy to see that the $p$th part of $\uB^{(i)}, \uB_p^{(i)}  = \mathbf{0}$. This implies that if the $p$th part of $\uB^{(i)}$ associated with $\lambda^{(i)}>1$, $\uB_p^{(i)} \neq \mathbf{0}$, then the $i$th signal component in the $p$th data set is among the group of correlated $i$th components as desired. 
\end{proof}

Theorem \ref{theorem:evec} assumes that any eigenvalue of $\CB$ that is greater than one has an algebraic multiplicity of one. This is not a necessary but a sufficient condition. Appendix \ref{app2} discusses its sufficiency and also scenarios in which the correlation structure can be completely determined using Theorem \ref{theorem:evec} even when this assumption is not true. 

As a result of Theorems \ref{theorem:eval} and \ref{theorem:evec}, if the $i$th eigenvalue of $\CB$ is greater than one and is unique as an eigenvalue, the existence of correlation associated with the $i$th component of the $p$th data set can be determined by testing the hypothesis $\mathbf{u}^{(i)}_p = \mathbf{0}$. A bootstrap-based hypothesis test for this purpose is proposed in Section \ref{sec:Evec_test}.  

\section{Bootstrap-based Tests for Eigenvalues and Eigenvectors of  $\CB$}
\label{sec:bt_ht}
\subsection{Test for detecting the eigenvalues of $\CB$ representing the correlated subspace}
\label{sec:Eval_test}

Theorem 1 allows us to infer that the number of eigenvalues of $\CB$ greater than one is equal to the dimension of the correlated subspace, $d$. In practice, however, the composite coherence matrix, $\CB$, is unknown and has to be estimated from the samples. As a result, the number of eigenvalues of the sample composite coherence matrix that are greater than one will often not equal the dimension of the correlated subspace, $d$. This inconsistency is addressed in the related model-order selection literature by setting a threshold for the eigenvalues that is determined with an information-theoretic criterion (ITC) or hypothesis testing. The algorithm presented here fits in the category of  hypothesis testing. 

A popular approach to solve the problem of model-order selection using hypothesis testing is to perform a series of binary hypothesis tests and stop when a certain condition is met. In this context, this means starting with a counter $s = 0$ and performing the following binary test of null hypothesis $H_0$ and alternative $H_1$
\begin{align}
\label{eq:btest}
H_0 &: \text{{$d=s$}}, \notag \\
H_1 &: \text{{$d>s$}}.
\end{align}
If $H_0$ is rejected, $s$ is incremented and another test of $H_0$ vs. $H_{1}$ is run. This is repeated until $H_0$ is retained or $s$ reaches its maximum possible value.

The binary test in \eqref{eq:btest} requires a statistic whose (asymptotic) distribution under $H_0$ is theoretically known or estimated from the samples. The distribution of the statistic can be derived assuming a particular correlation structure among the components \cite{hasija2016detecting}, but for arbitrary correlation structures the distribution is unknown. In this work, we use the bootstrap technique to estimate this distribution.    

Let the eigenvalues of $\CB$ be arranged with ordering ${\lambda}^{(1)} \geq {\lambda}^{(2)}  \ldots \geq {\lambda}^{(nP)}$. To estimate $d$, we propose a statistic that is based on testing whether the eigenvalues of $\CB$ following its $d$ largest eigenvalues are equal to one. If a signal component in one data set is not correlated with the signal components of any other data set, at least one eigenvalue of C is one. In order to increase the power and stability of the test, we make the stronger assumption that each data set has at least one signal component that is completely uncorrelated in this manner, i.e., $d_{pq} < n$ for all $p$,$q$. Therefore, $\CB$ has at least $P$ eigenvalues equal to one and the null hypothesis for each test is
\begin{equation}
\label{eq:nullhyp}
H_0 : \lambda^{(s+1)} = \lambda^{(s+2)} = \ldots  = \lambda^{(s+P)} = 1. 
\end{equation}
Note that we cannot include all $nP-d$ eigenvalues following the $d$ largest eigenvalues of $\CB$ since some of them are less than one. 
The null hypothesis is rejected when the proposed test statistic,
\begin{equation}
\label{eq:Tstatistic}
T(s) = \sum_{i=s+1}^{s+P} (\lambda^{(i)} - 1)^2  , 
\end{equation}
is sufficiently far from zero. The distribution of $T(s)$ under $H_0$ is estimated using the bootstrap \cite{tubiblio66457}.

In Algorithm \ref{alg:eval_test}, we present an algorithm for estimating $d$ using $T{(s)}$ given $M$  joint samples of data sets $\xB_1, \ldots, \xB_P$.  
The $M$ samples of each data set form the columns of the data matrices $\mathbf{X}_1,\ldots,\mathbf{X}_P$. The $P$ data matrices along with the number of bootstrap resamples, $B$, and the probability of false alarm, $P_\text{fa}$, are the inputs of Algorithm \ref{alg:eval_test}.  

\begin{algorithm} 
	\caption{Estimator for the dimension of the correlated subspace of $P$ data sets} 
	\label{alg:eval_test}
	\begin{spacing}{1.3}
	\begin{algorithmic}[1]
	{\fontsize{9.2}{9.2}\selectfont
	 \State \textbf{Input} $\big\{\tensor*[]{\mathbf{X}}{_p}\big\}_{p=1}^P$ : observations \vskip0.1ex \text{  }$B$: number of bootstrap resamples  \vskip0.1ex  \text{  }$P_\text{fa}$: probability of false alarm
	  \State \textbf{Output} $\hat{d}$: dimension of correlated subspace
	  \Function{\textsc{CorrDim}}{$\big\{\tensor*[]{\mathbf{X}}{_p}\big\}_{p=1}^P,B,\tensor*[]{P}{_{\textrm{fa}}}$} 
	  
	  	  \State $\mathbf{\widehat{R}}_{D} \gets \frac{1}{M} \textrm{blkdiag}\big(  \tensor*[]{\mathbf{X}}{_{1}}   \tensor*[]{\mathbf{X}}{_1^T},  \ldots ,\tensor*[]{\mathbf{X}}{_{P}}\tensor*[]{\mathbf{X}}{_{P}^T}\big)$ 

	   \State $\mathbf{X} \gets \big[ \tensor*[]{\mathbf{X}}{_1^T}, \ldots , \tensor*[]{\mathbf{X}}{_P^T}\big]^T$ 
	   
	   \State $ \mathbf{\widehat{R}} \gets \frac{1}{M} \mathbf{X} \mathbf{X}^T$
	   
	   \State $\mathbf{\widehat{C}} \gets	  \tensor*[]{\mathbf{\widehat{R}}}{_{D} ^{\nicefrac{-1}{2}}}  \mathbf{\widehat{R}}  \tensor*[]{\mathbf{\widehat{R}}}{_{D} ^{\nicefrac{-1}{2}}} $ 
	   
	   \State $\pmb{\hat{\lambda}} \gets \textrm{eigenvalues}\big(\mathbf{\widehat{C}}\big)$ \Comment{s.t. $\tensor*[]{\hat{\lambda}}{^{(1)}} \geq \cdots \geq \tensor*[]{\hat{\lambda}}{^{(nP)}}$}
	   
	   \For{$b=1,\ldots,B$} \Comment{bootstrap resamplings indexed by left subscript}
	   
	   \For{$m=1,\ldots,M$}
	   
	   \State $\tensor*[_{b}]{j}{_m} \gets \textrm{random integer } [1,M]$ \Comment{resampling indices chosen with replacement}	
	   \EndFor
	   
	   \For{$p=1,\ldots,P$}
	   
	   \State $\tensor*[_b]{\mathbf{X}}{_{p}} \gets \big[ \tensor*[]{\mathbf{x}}{_p}(\tensor*[_{b}]{j}{_1}), \ldots,\tensor*[]{\mathbf{x}}{_p}(\tensor*[_{b}]{j}{_M})\big]$  

	   \EndFor
	   
	   \State $\tensor*[_{b}]{\mathbf{\widehat{R}}}{_{D}} \gets \frac{1}{M} \textrm{blkdiag}\big( \tensor*[_{b}]{\mathbf{X}}{_1} \tensor*[_{b}]{\mathbf{X}}{_1^T} ,  \ldots , \tensor*[_{b}]{\mathbf{X}}{_P}\tensor*[_{b}]{\mathbf{X}}{_P^T} \big)$ 
	   
	   \State $\tensor*[_{b}]{\mathbf{X}}{}  \gets \big[\tensor*[_{b}]{\mathbf{X}}{_1^T} , \ldots , \tensor*[_{b}]{\mathbf{X}}{_P^T} \big]^T$ 
	   	   
	   \State $ \tensor*[_{b}]{\mathbf{\widehat{R}}}{} \gets \frac{1}{M} \tensor*[_{b}]{\mathbf{X}}{}  \tensor*[_{b}]{\mathbf{X}}{^T} $
	   
	   \State $\tensor*[_{b}]{\mathbf{\widehat{C}}}{} \gets	  \tensor*[_{b}]{\mathbf{\widehat{R}}}{_{D}^{\nicefrac{-1}{2}}}  \tensor*[_{b}]{\mathbf{\widehat{R}}}{}   \tensor*[_{b}]{\mathbf{\widehat{R}}}{_{D}^{\nicefrac{-1}{2}}} $ 
	   
	   \State $\tensor*[_{b}]{\pmb{\hat{\lambda}}}{} \gets \textrm{eigenvalues}\big( \tensor*[_{b}]{\mathbf{\widehat{C}}}{} \big)$ \Comment{s.t. $ \tensor*[_{b}]{\hat{\lambda}}{^{(1)}} \geq  \cdots \geq \tensor*[_{b}]{\hat{\lambda}}{^{(nP)}} $}
	   
	   \EndFor
	   
	   \State $s_{\max} \gets n-1$
	   
		\For{$s=0,\ldots,s_{\max}$}
		
		\State $T(s) \gets \sum_{i=s+1}^{s+P} (\hat{\lambda}^{(i)}-1)^2$ \Comment{compute test statistic for $H_0: d = s$}
		
		\For{$b=1,\ldots,B$}
		
		\State $\tensor*[_{b}]{T}{}(s) \gets \sum_{i=s+1}^{s+P} \big(\tensor*[_{b}]{\hat{\lambda}}{^{(i)}}-1\big)^2$ \Comment{estimate distribution of $T(s)$ under $H_0$}
		
		\EndFor
		\State $P(s) \gets \frac{1}{B} \sum_{b=1}^B I\big( T(s) \leq | \tensor*[_{b}]{T}{}(s)  - T(s)| \big)$ \Comment{$I(\cdot)$ is the indicator function}
		
		\EndFor \\

		\Return  $\displaystyle \hat{d} \gets \min \big\{\argmin_{s=0,\ldots,s_{\max}} P(s)  \geq \tensor*[]{P}{_{\textrm{fa}}}, n-1 \big\}$
		
		\EndFunction }
	\end{algorithmic}
	\end{spacing} 
\end{algorithm}

\subsection{Test for eigenvectors of $\CB$  corresponding to correlated components} 
\label{sec:Evec_test}
In addition to identifying the dimension of the correlated subspace, $d$, our stated goal is to estimate the structure of the correlations between the collection of data sets. As a result of  Theorem \ref{theorem:evec}, we need to identify the values of $i$ and $p$ for which the subvector $\uB_{p}^{(i)}= \mathbf{0}$ in order to determine which data sets have an uncorrelated $i$th signal component. However, we still do not have direct access to the composite coherence matrix.  When $\CB$ is estimated from samples, these subvectors will not be exactly zero. Thus we propose a novel method for identifying multiset correlation structure that uses a bootstrap-based test to detect zero subvectors in the eigenvectors of $\CB$. 

Assuming ${d}$ correlated components (which can be estimated via Algorithm \ref{alg:eval_test}), the technique tests the ${d}$ eigenvectors associated with the $d$ largest eigenvalues of the sample composite coherence matrix. For $i = 1 \ldots {d}$ and $p = 1 \ldots P$ we test the hypotheses
\begin{align}
\label{eq:evec_btest}
H_0 &: \text{{$\mathbf{u}_p^{(i)} = \mathbf{0}$}}, \notag \\
H_1 &: \text{{$\mathbf{u}_p^{(i)} \neq \mathbf{0}$}}.  
\end{align}
The null hypothesis is rejected when the squared Euclidean norm of $\mathbf{u}_p^{(i)}$, 
\begin{equation}
\label{eq:Tstatistic_evec}
T = \| {\uB}_p^{(i)} \|^2, 
\end{equation} 
is sufficiently far from zero. The distribution of the statistic $T$ under the null hypothesis is estimated using a bootstrap akin to the one in Algorithm \ref{alg:eval_test}. The full procedure is given in Algorithm \ref{alg:evec_test}. 

\begin{algorithm} 
	\caption{Estimator for the correlation structure of $P$ data sets} 
	\label{alg:evec_test}
	\begin{spacing}{1.3}
	\begin{algorithmic}[1]
	{\fontsize{9.2}{9.2}\selectfont
	\State \textbf{Input} $\big\{\tensor*[]{\mathbf{X}}{_p}\big\}_{p=1}^P$ : observations \vskip0.1ex \text{  }$\hat{d}$: dimension of correlated subspace \vskip0.1ex \text{  }$B$: number of bootstrap resamples  \vskip0.1ex  \text{  }$P_\text{fa}$: probability of false alarm 
	  \State \textbf{Output} $\widehat{\mathbf{Z}}$ : correlation map
	  
	  \Function{\textsc{CorrStruc}}{$\left\{\mathbf{X}_p\right\}_{p=1}^P,\hat{d},B,P_{\textrm{fa}}$} 
	  
	  \State $\widehat{\mathbf{Z}} \gets [ \mathbf{1} ] \in \mathbb{R}^{\hat{d} \times {P\choose 2}}$  \Comment{rows indexed by signal components, columns by pairs of data sets in lexicographical order}
	  
	  	  \State $\mathbf{\widehat{R}}_{D} \gets \frac{1}{M} \textrm{blkdiag}\big(  \tensor*[]{\mathbf{X}}{_{1}}   \tensor*[]{\mathbf{X}}{_1^T},  \ldots ,\tensor*[]{\mathbf{X}}{_{P}}\tensor*[]{\mathbf{X}}{_{P}^T}\big)$ 

	   \State $\mathbf{X} \gets \big[ \tensor*[]{\mathbf{X}}{_1^T},  \ldots , \tensor*[]{\mathbf{X}}{_P^T}\big]^T$ 
	   
	   \State $ \mathbf{\widehat{R}} \gets \frac{1}{M} \mathbf{X} \mathbf{X}^T$
	   
	   \State $\mathbf{\widehat{C}} \gets	  \tensor*[]{\mathbf{\widehat{R}}}{_{D} ^{\nicefrac{-1}{2}}}  \mathbf{\widehat{R}}  \tensor*[]{\mathbf{\widehat{R}}}{_{D} ^{\nicefrac{-1}{2}}} $ 
	   \For{$i=1,\ldots,\hat{d}$}
	   \State $ \mathbf{\hat{u}}^{(i)} \gets \textrm{eigenvector}(\mathbf{\widehat{C}})$ \Comment{ordered by associated eigenvalue s.t. $\tensor*[]{\hat{\lambda}}{^{(1)}} \geq  \cdots \geq \tensor*[]{\hat{\lambda}}{^{(\hat{d})}}$} \vskip0.3ex
	   \State $\tensor*[]{\mathbf{\hat{u}}}{^{(i)}} =\left[\tensor*[]{\mathbf{\hat{u}}}{_1^{(i)T}}, \ldots, \tensor*[]{\mathbf{\hat{u}}}{_P^{(i)T}}  \right]^T \textrm{ with } \tensor*[]{\mathbf{\hat{u}}}{_p^{(i)}}  \in \mathbb{R}^n \hspace{.25em} \forall p$ 
	   \EndFor

 \For{$b=1,\ldots,B$} \Comment{bootstrap resamplings indexed by left subscript}
	   
	   \For{$m=1,\ldots,M$}
	   
	   \State $\tensor*[_{b}]{j}{_m} \gets \textrm{random integer } [1,M]$ \Comment{resampling indices chosen with replacement}	
	   \EndFor
	   
	   \For{$p=1,\ldots,P$}
	   
	   \State $\tensor*[_b]{\mathbf{X}}{_{p}} \gets \big[ \tensor*[]{\mathbf{x}}{_p}(\tensor*[_{b}]{j}{_1}), \ldots,\tensor*[]{\mathbf{x}}{_p}(\tensor*[_{b}]{j}{_M})\big]$ 

	   \EndFor
	   
	   \State $\tensor*[_{b}]{\mathbf{\widehat{R}}}{_{D}} \gets \frac{1}{M} \textrm{blkdiag}\big( \tensor*[_{b}]{\mathbf{X}}{_1} \tensor*[_{b}]{\mathbf{X}}{_1^T} ,  \ldots , \tensor*[_{b}]{\mathbf{X}}{_P}\tensor*[_{b}]{\mathbf{X}}{_P^T} \big)$ 
	   
	   \State $\tensor*[_{b}]{\mathbf{X}}{}  \gets \big[\tensor*[_{b}]{\mathbf{X}}{_1^T} , \ldots , \tensor*[_{b}]{\mathbf{X}}{_P^T} \big]^T$ 
	   
	   \State $ \tensor*[_{b}]{\mathbf{\widehat{R}}}{} \gets \frac{1}{M} \tensor*[_{b}]{\mathbf{X}}{}  \tensor*[_{b}]{\mathbf{X}}{^T} $
	   
	   \State $\tensor*[_{b}]{\mathbf{\widehat{C}}}{} \gets	  \tensor*[_{b}]{\mathbf{\widehat{R}}}{_{D}^{\nicefrac{-1}{2}}}  \tensor*[_{b}]{\mathbf{\widehat{R}}}{}   \tensor*[_{b}]{\mathbf{\widehat{R}}}{_{D}^{\nicefrac{-1}{2}}} $
	   
	   	   \For{$i=1,\ldots,\hat{d}$}
	   \State $ \tensor*[_{b}]{\mathbf{\hat{u}}}{^{(i)}}  \gets \textrm{eigenvector}(\tensor*[_{b}]{\mathbf{\widehat{C}}}{})$ \Comment{ordered by associated eigenvalue s.t. $ \tensor*[_{b}]{\hat{\lambda}}{^{(1)}} \geq \cdots \geq \tensor*[_{b}]{\hat{\lambda}}{^{(\hat{d})}}$ } 
	   \State {
	    $\tensor*[_b]{\mathbf{\hat{u}}}{^{(i)}} =\left[\tensor*[_b]{\mathbf{\hat{u}}}{_1^{(i)T}}, \ldots, \tensor*[_b]{\mathbf{\hat{u}}}{_P^{(i)T}}  \right]^T \textrm{with } \tensor*[_b]{\mathbf{\hat{u}}}{_p^{(i)}}  \in \mathbb{R}^n \hspace{.25em} \forall p$} 
	   \EndFor 
	   
	   \EndFor
	   
	   \For{$i = 1, \ldots, \hat{d}$}
	   \For{$p = 1, \ldots, P$}
	   \State $\tensor*[]{T}{} \gets \| \tensor*[]{\mathbf{\hat{u}}}{_p^{(i)}} \|^2$ \Comment{compute test statistic for $H_0: \| \tensor*[]{\mathbf{\hat{u}}}{_p^{(i)}} \|^2 = 0$}
	   \For{$b=1,\ldots,B$}
	   \State $\tensor*[_b]{T}{} \gets \| \tensor*[_b]{\mathbf{\hat{u}}}{_p^{(i)}} \|^2$ \Comment{estimate distribution of $\tensor*[]{T}{} $ under $H_0$}
	   \EndFor
	   \State $\tensor*[]{P}{} \gets \frac{1}{B} \sum_{b=1}^B I\big( \tensor*[]{T}{} \leq | \tensor*[_b]{T}{}  - \tensor*[]{T}{}| \big)$ \Comment{$I(\cdot)$ is the indicator function}
	   	   \If{$\tensor*[]{P}{} \geq \tensor*[]{P}{_{\textrm{fa}}}$}
	   \State $\mathbf{\widehat{Z}}(i,j\{p,q\}) \gets 0 \hspace{.25em} \forall q$ \Comment{ $j\{p,q\}$ gets linear index of $p,q$ in lexicographical order}
	   
	   \EndIf  
	   \EndFor
	   \EndFor   
	    \\
		\Return  $\mathbf{\widehat{Z}}$
		\EndFunction }
	\end{algorithmic}
	\end{spacing}
\end{algorithm}

Combining Algorithms \ref{alg:eval_test} and \ref{alg:evec_test} leads to an effective method for determining the correlation structure among multiple data sets. Algorithm \ref{alg:eval_test} determines how many signal components, ${d}$, have nonzero correlations, and Algorithm \ref{alg:evec_test} reveals the data sets across which these ${d}$ components are correlated. The final output is a binary matrix, $\mathbf{Z}$, that is similar to Tables \ref{tab:Example_1} and \ref{tab:Example_2} except that nonzero correlation coefficients are represented by ones. We refer to this as a correlation map, an example of which can be seen in Fig. \ref{fig:corr_struc_true}.

\section{Numerical Results} 
\label{sec:results}

In this section, we use Monte Carlo simulations to demonstrate the performance of the proposed technique that combines Algorithms \ref{alg:eval_test} and \ref{alg:evec_test}. Initially, we compare our technique with those reported by \cite{wu2002determination,hasija2016detecting,hasija2016bootstrap,bhinge2017estimation} which aim to estimate the number of components correlated across all data sets, $d_\text{all}$. To estimate $d_\text{all}$ for our technique, we run Algorithms \ref{alg:eval_test} and \ref{alg:evec_test} and then count the number of components that are correlated across all the data sets.  Next, we investigate the behavior of the proposed technique when the pairwise correlation coefficients are not above the threshold necessary for the proof of Theorem \ref{theorem:eval}. We show that the method is robust to the violation of this assumption and that the accuracy remains high for many correlation structures. Finally, we compare the performance of our technique with an approach in \cite{marrinancomplete}, which also determines the complete correlation structure in multiple data sets. This comparison highlights the quantitative and qualitative differences between the two techniques. 

We present results with four different correlation structures. Some simulation settings are common to all four. The signal components in each data set are Gaussian distributed with variance of 1. The mixing matrices are randomly generated orthogonal matrices. Each data set is corrupted by additive white Gaussian noise (AWGN)\footnote{The Gaussian distribution of  signal components and noise, and the orthogonality of mixing matrices is not a requirement for the technique proposed in this work. MATLAB code for the proposed method and competing approaches can be found at 
\\ \url{https://github.com/SSTGroup/Correlation-Analysis-in-High-Dimensional-Data} and can be tested with different settings from the ones mentioned here.}. The variance of noise components is chosen according to the signal-to-noise-ratio (SNR) which is defined per component as
\begin{equation}
\text{SNR (dB)} = 10\log_{10} \bigg( \frac{\sigma_s^2}{\sigma_n^2} \bigg) , 
\end{equation}
where $\sigma_s^2$ is the variance of the signal and $\sigma_n^2$ is the variance of the noise. The SNR is the same for all data sets. The number of bootstrap resamples is $B=1000$ and the probability of false alarm is $P_{\text{fa}} = 0.05$. The performance plots are shown as a function of SNR, which is varied from $-10$ dB to $15$ dB. The results are averaged over $500$ independent trials. The performance of each method for determining model order is measured by the mean accuracy  (number of correct estimates divided by number of trials). The performance in estimating the complete correlation structure is measured using precision, 
i.e., the percentage of correctly detected correlations among all the detected correlations, and recall, i.e., the percentage of correctly detected correlations among all actual correlations. The four different scenarios are the following:   
\begin{enumerate}[label=\roman*)]
\item \textbf{Evaluation of model-order selection with special correlation structure}, for {$P=4$ data sets with $d = d_\text{all} = 3$}: Each data set is of dimension $n=7$ and the number of samples is $M = 350$. The components are either correlated across all data sets or are uncorrelated. Thus, the number of components that are correlated between at least a pair of data sets, $d$, is equal to the number of components correlated across all data sets,  $d_\text{all}$. This type of correlation structure satisfies the assumption in \cite{wu2002determination,hasija2016detecting}. The pairwise correlation coefficients for the three correlated components are shown in Table \ref{tab:Sim_eg1}, all of which exceed the $\epsilon^{(i)} = (3/4)^2$ threshold as assumed by Theorem \ref{theorem:eval}. 

Fig. \ref{fig:eval_scen1} shows the mean accuracy of $\hat{d}_\text{all}$ as a function of the SNR for the proposed and the competing techniques. All the techniques correctly estimate $d_\text{all}$ when the SNR is high. When the SNR is low, the proposed approach outperforms all competing techniques. 
\begin{table}
\begin{center}
 \begin{tabular}{ c | c c c c c c} 
 & $\rho_{12}^{(i)}$ & $\rho_{13}^{(i)}$ & $\rho_{14}^{(i)}$ & $\rho_{23}^{(i)}$ & $\rho_{24}^{(i)}$ & $\rho_{34}^{(i)}$ \\[1ex]
  \hline \Tstrut
 $i=1$ & 0.63 & 0.78 & 0.69 & 0.81 & 0.64 & 0.91 \\[.5ex]
 $i=2$ & 0.62 & 0.67 & 0.74 & 0.71 & 0.82 & 0.91 \\[.5ex]
 $i=3$ & 0.84 & 0.81 & 0.72 & 0.57 & 0.71 & 0.62 
\end{tabular}
\end{center}
\caption{ Correlation structure of the three correlated components in four data sets used in the first simulation setup.} 
\label{tab:Sim_eg1}
\end{table}

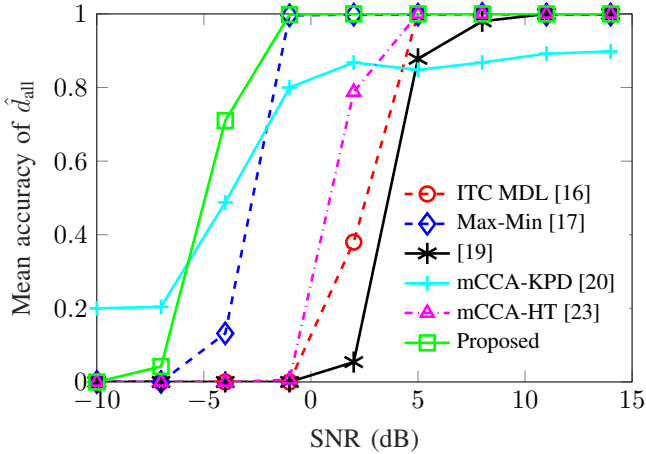
\begin{figure}[!t]
        \centering
        \setlength\fwidth{0.4\textwidth}
        \setlength\fheight{0.27\textwidth}
%
%
\definecolor{mycolor1}{rgb}{0.00000,1.00000,1.00000}%
\definecolor{mycolor2}{rgb}{1.00000,0.00000,1.00000}%
\begin{tikzpicture}

\begin{axis}[%
width=0.981\fwidth,
height=\fheight,
at={(0\fwidth,0\fheight)},
scale only axis,
xmin=-10,
xmax=15,
xlabel style={font=\color{white!15!black}},
xlabel={SNR (dB)},
ymin=0,
ymax=1,
ylabel style={font=\color{white!15!black}},
ylabel={Mean accuracy of $\hat{d}_\textrm{all}$},
axis background/.style={fill=white},
legend style={at={(.995,0.05)}, anchor=south east, legend cell align=left, align=left, draw=none, font=\fontsize{8.5}{8.5}\selectfont}
]
\addplot [color=red, dashed, line width=1.0pt, mark size=2.8pt, mark=o, mark options={solid, red}]
  table[row sep=crcr]{%
-10	0\\
-7	0\\
-4	0\\
-1	0\\
2	0.38\\
5	1\\
8	1\\
11	1\\
14	1\\
};
\addlegendentry{ITC MDL [16]}

\addplot [color=blue, dashed, line width=1.0pt, mark size=4pt, mark=diamond, mark options={solid, blue}]
  table[row sep=crcr]{%
-10	0\\
-7	0\\
-4	0.132\\
-1	0.996\\
2	0.998\\
5	0.998\\
8	1\\
11	0.998\\
14	0.998\\
};
\addlegendentry{Max-Min [17]}

\addplot [color=black, line width=1.0pt, mark size=4.0pt, mark=asterisk, mark options={solid, black}]
  table[row sep=crcr]{%
-10	0\\
-7	0\\
-4	0\\
-1	0\\
2	0.054\\
5	0.878\\
8	0.98\\
11	1\\
14	1\\
};
\addlegendentry{ [19]}

\addplot [color=mycolor1, line width=1.0pt, mark size=2.8pt, mark=+, mark options={solid, mycolor1}]
  table[row sep=crcr]{%
-10	0.2\\
-7	0.204\\
-4	0.488\\
-1	0.8\\
2	0.868\\
5	0.848\\
8	0.868\\
11	0.892\\
14	0.898\\
};
\addlegendentry{mCCA-KPD [20]}

\addplot [color=mycolor2, dashdotted, line width=1.0pt, mark size=2.8pt, mark=triangle, mark options={solid, mycolor2}]
  table[row sep=crcr]{%
-10	0\\
-7	0\\
-4	0\\
-1	0.006\\
2	0.788\\
5	1\\
8	1\\
11	1\\
14	1\\
};
\addlegendentry{mCCA-HT [23]}

\addplot [color=green, line width=1.0pt, mark size=2.8pt, mark=square, mark options={solid, green}]
  table[row sep=crcr]{%
-10	0\\
-7	0.042\\
-4	0.71\\
-1	0.998\\
2	1\\
5	1\\
8	1\\
11	1\\
14	1\\
};
\addlegendentry{Proposed}

\end{axis}
\end{tikzpicture}%
    \caption{Mean accuracy of $\hat{d}_\text{all}$ for the proposed and the competing techniques in detecting three components correlated across all four data sets.}
    \label{fig:eval_scen1}  
    \end{figure}

\item \textbf{Evaluation of model-order selection with arbitrary correlation structure}, for {$P=4$ data sets with $d = 3, d_\text{all} = 1$}: In this setting all the parameters are the same as in the previous scenario except that the first component of each data set is correlated  with all other data sets but the other two components are only correlated across a subset of data sets. The second component is correlated across all except the first data set and the third component is correlated between data sets two and four. The methods in \cite{wu2002determination,hasija2016detecting} are not evaluated as they are inapplicable in this setting. 

Fig. \ref{fig:eval_scen2} shows that the proposed technique works better than the techniques in \cite{hasija2016bootstrap,bhinge2017estimation,marrinancomplete} in estimating the model order $d_\text{all}$ for low values of SNR. It is also worth noting that the techniques in \cite{hasija2016bootstrap,bhinge2017estimation} estimate only $d_\text{all}$ while the proposed method and \cite{marrinancomplete} also detect the components correlated across subsets of the data sets along with their correlation structure.
 \begin{figure}[!t]
        \centering
        \setlength\fwidth{0.4\textwidth}
        \setlength\fheight{0.27\textwidth}
%
%
\definecolor{mycolor1}{rgb}{0.00000,1.00000,1.00000}%
\definecolor{mycolor2}{rgb}{1.00000,0.00000,1.00000}%
\begin{tikzpicture}

\begin{axis}[%
width=0.981\fwidth,
height=\fheight,
at={(0\fwidth,0\fheight)},
scale only axis,
xmin=-10,
xmax=15,
xlabel style={font=\color{white!15!black}},
xlabel={SNR (dB)},
ymin=0,
ymax=1,
ylabel style={font=\color{white!15!black}},
ylabel={Mean accuracy of $\hat{d}_\textrm{all}$},
axis background/.style={fill=white},
legend style={at={(.995,0.02)}, anchor=south east, legend cell align=left, align=left, draw=none, font=\fontsize{8.5}{8.5}\selectfont}
]
\addplot [color=black, line width=1.0pt, mark size=4pt, mark=asterisk, mark options={solid, black}]
  table[row sep=crcr]{%
-10	0\\
-7	0\\
-4	0\\
-1	0\\
2	0.002\\
5	0.596\\
8	0.994\\
11	1\\
14	1\\
};
\addlegendentry{[19]}

\addplot [color=mycolor1, line width=1.0pt, mark size=2.8pt, mark=+, mark options={solid, mycolor1}]
  table[row sep=crcr]{%
-10	0.302\\
-7	0.314\\
-4	0.478\\
-1	0.684\\
2	0.786\\
5	0.806\\
8	0.806\\
11	0.832\\
14	0.828\\
};
\addlegendentry{mCCA-KPD [20]}

\addplot [color=mycolor2, dashdotted, line width=1.0pt, mark size=2.8pt, mark=triangle, mark options={solid, mycolor2}]
  table[row sep=crcr]{%
-10	0\\
-7	0\\
-4	0\\
-1	0.266\\
2	0.902\\
5	1\\
8	1\\
11	1\\
14	1\\
};
\addlegendentry{mCCA-HT [23]}

\addplot [color=green, line width=1.0pt, mark size=2.8pt, mark=square, mark options={solid, green}]
  table[row sep=crcr]{%
-10	0.002\\
-7	0.088\\
-4	0.856\\
-1	0.928\\
2	0.934\\
5	0.952\\
8	0.966\\
11	0.972\\
14	0.978\\
};
\addlegendentry{Proposed}

\end{axis}
\end{tikzpicture}%
           \caption{ Mean accuracy of $\hat{d}_\text{all}$ for the proposed and the competing techniques in detecting $d_\text{all}=1$ component correlated across all four data sets, in presence of two signal components correlated across a subset of the data sets, i.e., $d=3$. }
    \label{fig:eval_scen2}  
    \end{figure}
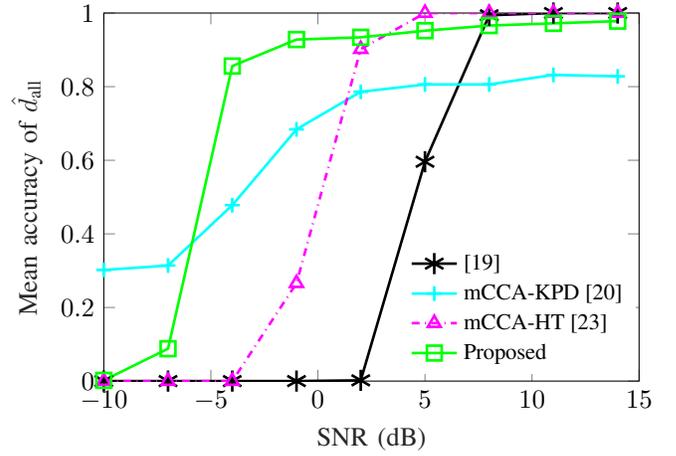
    
 \item \textbf{Performance of the proposed method when the element-wise threshold is not met}, for {$P=5$ data sets with $d = d_\text{all} = 2$}: 
We also investigate the performance of the proposed technique for determining the number of correlated components when some of the pairwise correlation coefficients do not meet the threshold required for Theorem \ref{theorem:eval}. 
For this, we assume that the first two components in each data set are correlated across all data sets. Thus, the threshold for the pairwise correlation coefficients given by Theorem \ref{theorem:eval} is $\epsilon^{(1)} = \epsilon^{(2)} = \epsilon = (\frac{4}{5})^2 = 0.64$. We keep some of the pairwise correlation coefficients above the threshold $ \epsilon$ and vary the remaining ones. More precisely, we set $\rho^{(i)}_{pq} = 0.7$ for $p > q$, $p = 2, 3, 4, 5$ and $q = 3, 4, 5$, which exceeds the threshold $\epsilon$, and jointly vary the remaining correlation coefficients $\rho_{12}^{(i)} = \rho_{13}^{(i)} = \rho_{14}^{(i)}= \rho_{15}^{(i)} = \rho$, for $i=1,2$. To demonstrate the relative robustness of the method against violating the assumption in Theorem \ref{theorem:eval}, we show the accuracy of $\hat{d}$ as a function of $\rho$ in Fig. \ref{fig:scen3} for different values of SNR. 

For $\rho < \epsilon$, the performance depends on the SNR: For low SNR, it becomes increasingly difficult to correctly estimate $d$ for only weakly correlated components. On the other hand, as long as the SNR is high enough, violating the threshold in Theorem \ref{theorem:eval} does not present a problem, and $d$ can still be correctly determined.
  \begin{figure}[!t]
        \centering
        \setlength\fwidth{0.38\textwidth}
        \setlength\fheight{0.27\textwidth}
%
%
\begin{tikzpicture}

\begin{axis}[%
width=\fwidth,
height=0.998\fheight,
at={(0\fwidth,0\fheight)},
scale only axis,
xmin=0.1,
xmax=0.9,
xlabel style={font=\color{white!15!black}},
xlabel={Correlation coefficient $\rho$},
ymin=0,
ymax=1,
ylabel style={font=\color{white!15!black}},
ylabel={Mean accuracy of $\hat{d}$},
axis background/.style={fill=white},
legend style={at={(0.197,0.362)}, anchor=south west, legend cell align=left, align=left, draw=none, font=\fontsize{9}{9}\selectfont}
]
\addplot [color=black, dotted, line width=1.2pt, mark size=3.2pt, mark=diamond, mark options={solid, black}]
  table[row sep=crcr]{%
0.88	0.996\\
0.8	0.996\\
0.7	0.98\\
0.6	0.952\\
0.5	0.912\\
0.4	0.858\\
0.3	0.776\\
0.2	0.71\\
0.1	0.672\\
};
\addlegendentry{SNR = -2.5dB}

\addplot [color=green, dashdotted, line width=1.2pt, mark size=2.8pt, mark=o, mark options={solid, green}]
  table[row sep=crcr]{%
0.88	1\\
0.8	1\\
0.7	1\\
0.6	1\\
0.5	1\\
0.4	1\\
0.3	1\\
0.2	1\\
0.1	1\\
};
\addlegendentry{SNR = 0dB}

\addplot [color=red, dashed, line width=1.2pt]
  table[row sep=crcr]{%
0.64	0\\
0.64	1\\
};

\node[below right, align=left, font=\color{red}]
at (rel axis cs:0.453,0.212) {$\epsilon = 0.64$};
\end{axis}
\end{tikzpicture}%
           \caption{ Mean accuracy for estimating $d = 2$ components correlated in five data sets as a function of the correlation coefficient $\rho$.}
    \label{fig:scen3}  
    \end{figure}
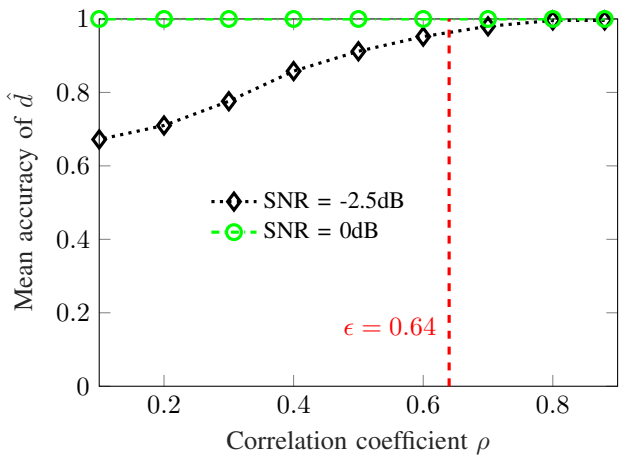    
     
\item \textbf{Evaluation of complete correlation structure}, for {$P=5$ data sets with $d = 3$, $d_\text{all} = 1$}: Finally, we compare the performance of the proposed method and the technique in \cite{marrinancomplete} for determining the complete correlation structure. Each data set has $n=4$ components and the number of samples is $M=250$. Of the $d=3$ correlated components, $d_\text{all} = 1$ and the first component of each data set is correlated with the first component of each other data set. The second component of each data set is correlated across all the data sets except the fourth data set. Finally, the third components of data sets one, four and five are correlated. Each pairwise correlation coefficient is $0.7$, thus exceeding the threshold as required in Theorem \ref{theorem:eval}. 
 
 Fig. \ref{fig:mean_acc_d} shows the mean accuracy of estimating ${d}$ for both techniques as a function of SNR. Both methods perform well for medium and high SNR values when estimating $d$, with the proposed method slightly outperforming the method in \cite{marrinancomplete}  for low values of SNR.  Fig. \ref{fig:mean_pr} shows the mean precision and recall values for determining the complete correlation structure. The precision of the method in \cite{marrinancomplete} is better at low SNR while the recall for the proposed method is better at both low and medium SNRs. 
 
 \begin{figure}[h]
    \centering
    \begin{subfigure}[b]{0.5\textwidth}
        \centering
        \setlength\fwidth{0.73\textwidth}
        \setlength\fheight{0.52\textwidth}
%
%
\definecolor{mycolor1}{rgb}{1.00000,0.00000,1.00000}%
\begin{tikzpicture}

\begin{axis}[%
width=\fwidth,
height=0.992\fheight,
at={(0\fwidth,0\fheight)},
scale only axis,
xmin=-10,
xmax=15,
xlabel style={font=\color{white!15!black}},
xlabel={SNR (dB)},
ymin=0,
ymax=1,
ylabel style={font=\color{white!15!black}},
ylabel={Mean accuracy of $\hat{d}$},
axis background/.style={fill=white},
legend style={at={(.99,0.12)}, anchor=south east, legend cell align=left, align=left, draw=none,font=\fontsize{9}{9}\selectfont}
]
\addplot [color=mycolor1, dotted, line width=1.2pt, mark size=2.8pt, mark=triangle, mark options={solid, mycolor1}]
  table[row sep=crcr]{%
-10	0\\
-7	0.006\\
-4	0.202\\
-1	0.804\\
2	0.976\\
5	0.978\\
8	0.994\\
11	0.988\\
14	0.996\\
};
\addlegendentry{mCCA-HT [23]}

\addplot [color=green, line width=1.2pt, mark size=2.8pt, mark=square, mark options={solid, green}]
  table[row sep=crcr]{%
-10	0.002\\
-7	0.032\\
-4	0.322\\
-1	0.8\\
2	1\\
5	1\\
8	1\\
11	1\\
14	1\\
};
\addlegendentry{Proposed}

\end{axis}
\end{tikzpicture}%
        \caption{}
        \label{fig:mean_acc_d}
    \end{subfigure}
    \begin{subfigure}[b]{0.5\textwidth}
        \centering
        \setlength\fwidth{0.73\textwidth}
        \setlength\fheight{0.52\textwidth}
%
%
\definecolor{mycolor1}{rgb}{1.00000,0.00000,1.00000}%
\begin{tikzpicture}

\begin{axis}[%
width=\fwidth,
height=0.986\fheight,
at={(0\fwidth,0\fheight)},
scale only axis,
xmin=-10,
xmax=15,
xlabel style={font=\color{white!15!black}},
xlabel={SNR (dB)},
ymin=0,
ymax=1,
ylabel style={font=\color{white!15!black}},
ylabel={Mean Precision and Recall},
axis background/.style={fill=white},
legend style={at={(.995,0.05)}, anchor=south east, legend cell align=left, align=left, draw=none,font=\fontsize{8.5}{8.5}\selectfont}
]
\addplot [color=mycolor1, dotted, line width=1pt, mark size=2.5pt, mark=triangle, mark options={solid, mycolor1}]
  table[row sep=crcr]{%
-10	0.13\\
-7	0.517366666666667\\
-4	0.937600793650794\\
-1	0.968316199091586\\
2	0.974542900142514\\
5	0.981167418546366\\
8	0.983663157894737\\
11	0.985444360902256\\
14	0.983134837092732\\
};
\addlegendentry{Precision mCCA-HT [23]}

\addplot [color=green, line width=1pt, mark size=2.5pt, mark=square, mark options={solid, green}]
  table[row sep=crcr]{%
-10	0.000933333333333333\\
-7	0.357662092585777\\
-4	0.920432960431457\\
-1	0.964284676269574\\
2	0.991936779162866\\
5	0.995581027667984\\
8	0.994920036485254\\
11	0.996399209486166\\
14	0.99195378534509\\
};
\addlegendentry{Precision Proposed}

\addplot [color=mycolor1, dotted, line width=1pt, mark size=2.5pt, mark=o, mark options={solid, mycolor1}]
  table[row sep=crcr]{%
-10	0.0076842105263158\\
-7	0.0410526315789474\\
-4	0.246842105263158\\
-1	0.725473684210526\\
2	0.969789473684212\\
5	0.982736842105264\\
8	0.985052631578948\\
11	0.987578947368421\\
14	0.984105263157895\\
};
\addlegendentry{Recall mCCA-HT [23]}

\addplot [color=green, line width=1pt, mark size=2.5pt, mark=o, mark options={solid, green}]
  table[row sep=crcr]{%
-10	0.000736842105263158\\
-7	0.198947368421053\\
-4	0.794631578947371\\
-1	0.955894736842107\\
2	1\\
5	1\\
8	1\\
11	1\\
14	1\\
};
\addlegendentry{Recall Proposed}

\end{axis}
\end{tikzpicture}%
        \caption{}
        \label{fig:mean_pr}
    \end{subfigure}
    \caption{Performance of the proposed technique and the technique in \cite{marrinancomplete} for determining the complete correlation structure in five data sets.  
    a) Mean accuracy of estimating $d$, the total number of correlated signal components b) Precision and recall for determining the complete correlation structure of the detected components. }
    \label{fig:mean_acc_wu_true}  
    \end{figure}
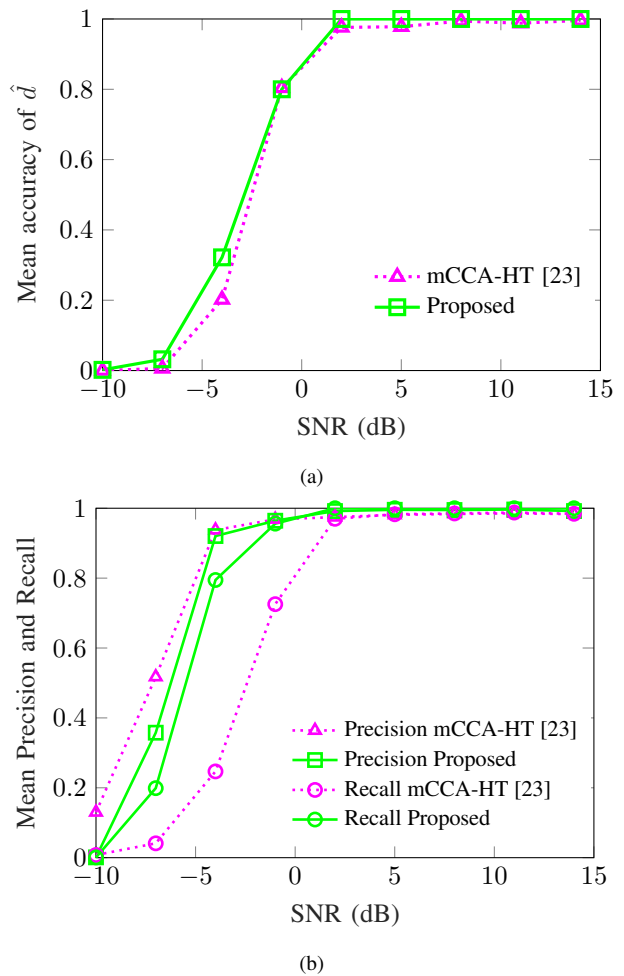
    
 The main reason for this difference is tied to the information that each method uses for hypothesis testing. The method in \cite{marrinancomplete} performs mCCA to extract the sets of signal components that are highly correlated via a deflationary approach. The components are extracted jointly from all data sets. Then, hypothesis tests are conducted on the extracted signal components from pairs of data sets to detect the underlying correlation structure. At low SNR values some of the correlations are missed while testing an individual pair of data sets, so its average recall is small in this regime. 
 
 On the other hand, the proposed method detects the components and their correlation structure by applying hypothesis tests to the eigenvalues and eigenvectors of the composite coherence matrix directly. The eigenvalues and the eigenvectors are functions of \textit{all} the pairwise correlation coefficients associated with the component and thus, the proposed method tests on this joint information. 
 
 This is further illustrated in Fig. \ref{fig:corr_struc_est}, which shows heat maps of the average accuracy for the two methods in estimating the complete correlation structure for this scenario. The true correlation structure is visualized in Fig. \ref{fig:corr_struc_true} using a binary map. This map mirrors the structure of  Tables \ref{tab:Example_1}, \ref{tab:Example_2} and \ref{tab:Sim_eg1} but represents the nonzero correlation coefficients with white blocks and the zero correlation coefficients with black blocks.  This binary map can be compared to heat maps of the simulation results to assess qualitative differences between the proposed method and the technique in \cite{marrinancomplete}. Ideally, these heat maps should look exactly like the binary map in Fig. \ref{fig:corr_struc_true}.   
     \begin{figure}[!t]
        \centering
        \includegraphics[scale=0.32]{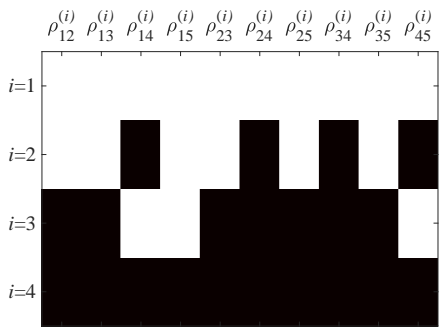}
           \caption{The correlation structure of four components correlated in five data sets. The white blocks represent nonzero correlation coefficients and the black blocks represent zero correlation coefficients.} 
    \label{fig:corr_struc_true}  
    \end{figure}    
    \begin{figure*}[!t]
        \centering
        \includegraphics[scale=0.285]{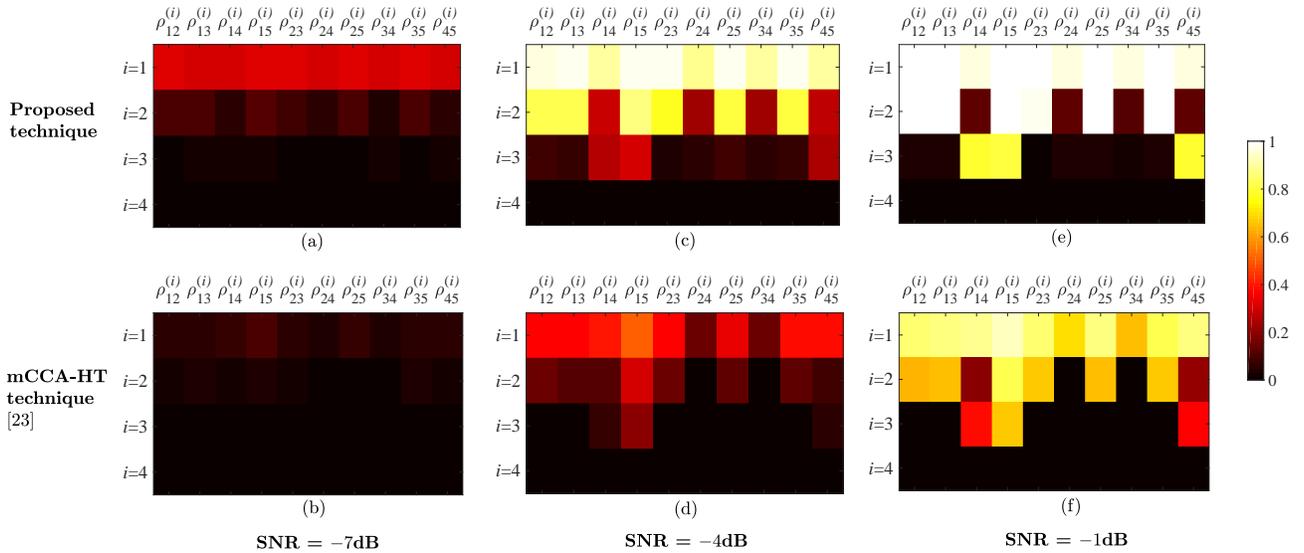}
           \caption{Heat maps showing the mean accuracy of detecting individual correlations for the proposed method and the method in \cite{marrinancomplete} at three different SNR values of  $-7$dB, $-4$dB and $-1$dB. The true correlation structure is shown in Fig. \ref{fig:corr_struc_true}.} 
    \label{fig:corr_struc_est}  
    \end{figure*}
In the low SNR regime (shown in Fig. \ref{fig:corr_struc_est}a and b at SNR = $-7$dB), both techniques detect very few correlations as illustrated by the dark color of the boxes. Some of the boxes in the second row corresponding to $i=2$ in Fig. \ref{fig:corr_struc_est}a that should be black are dark brown, indicating that the proposed method detects a few false positives, and therefore has low precision value compared to the method in \cite{marrinancomplete}. 

The main differences between performance of the two methods start to appear from SNR of $-4$dB. Fig. \ref{fig:corr_struc_est}c,d and \ref{fig:corr_struc_est}e,f show the heat maps for SNR$=-4$dB and SNR$=-1$dB, respectively. In Fig. \ref{fig:corr_struc_est}d and \ref{fig:corr_struc_est}f, boxes corresponding to the nonzero correlations of same component number have different colors.  This is expected because \cite{marrinancomplete} conducts tests on pairs of data sets so it is possible to detect the correlation between one pair and miss it between another.  On the other hand, the proposed technique tests the eigenvalues and eigenvectors of the composite coherence matrix, which provides a summary of all the pairwise correlations. Therefore, the boxes of the nonzero correlations corresponding to components with same component number are more uniform in color in Fig. \ref{fig:corr_struc_est}c and \ref{fig:corr_struc_est}e.  

At low SNR values, the number of data sets correlated across a particular component affects the accuracy of the proposed method.  The more data sets that are correlated along a given component, the better the proposed method performs. This can be observed in Fig. \ref{fig:corr_struc_est}c. The boxes of the first row ($i=1$) corresponding to nonzero correlations are significantly brighter than those of second and third rows ($i=2,3$), indicating a higher accuracy when detecting the first component. Similarly in Fig. \ref{fig:corr_struc_est}e, there is a contrast between the high accuracy when detecting the correlations of first and second components and the relatively lower accuracy of detecting the third component. This is because the eigenvalue associated with the component correlated across more data sets is significantly greater than one and thus makes its detection possible even when the noise power is high. This stands in contrast with the method from \cite{marrinancomplete} where no advantage is gained by the number of data sets across which a component is correlated. In Fig. \ref{fig:corr_struc_est}d and  \ref{fig:corr_struc_est}f we can see that boxes corresponding to the nonzero correlations of $i=1,2,3$ have less variation in color, indicating similar accuracy across the board.         
  
\end{enumerate}

\section{Conclusion} 
We have presented a technique that solves a model-selection problem to determine the complete correlation structure in multiple sets of data, i.e., identifying which components are correlated across which data sets. We proved that the eigenvalues and the eigenvectors of the composite coherence matrix completely characterize this information. Using these results, we developed a bootstrap-based hypothesis testing technique that estimates the complete correlation structure, which has shown competitive performance and broad applicability in various simulation scenarios. The proposed technique utilizes the joint information from all the data sets provided by the eigenvalues and the eigenvectors and thus performs well in detecting the components that are correlated across a large number of data sets even at low SNR. This is promising for various fields, e.g., biomedicine, genomics and remote sensing, where a common phenomenon is observed using different modalities.

\begin{appendix}
\label{app:lem1}
\subsection{Number of positive eigenvalues of a hollow symmetric matrix}
Crucial to the goal of identifying correlated signals via the eigenvalues of $\CB$ is a requirement that the correlated subspace of signal components associated with the same component number correspond to a single eigenvalue. As we see in Theorem \ref{theorem:eval}, this requirement necessitates identifying classes of matrices that have exactly one positive eigenvalue. The most general result in this domain comes from \cite{charles2013nonpositive} and characterizes the eigenvalues of hollow (zero-diagonal) symmetric nonnegative matrices. 

To leverage this result, we must first define a generalized Ramsey number. Let $\{ G_{a_1}, G_{a_2}, \ldots, G_{a_c}\}$ be a
collection of simple graphs where $a_i$ is the number of vertices of the $i$th graph. Suppose we wish to color the edges of a complete graph $G$ with $c$ colors. The \textit{generalized Ramsey number}, $R(G_{a_1}, G_{a_2}, \ldots, G_{a_c})$, is the minimum number of vertices of the complete graph $G$ such that for any $c$-coloring of $G$, there exists an $i \in \{ 1,2, \ldots, c\} $ such that $G_{a_i}$ is an induced subgraph of $G$ with all edges of color $i$. By Ramsey's Theorem \cite{ramsey2009problem} such a number  always exists. 

\begin{theorem}
\emph{(Charles, Farber, Johnsons, and Kennedy-Schaffer [\cite{charles2013nonpositive}, Theorem 3.5]).} Let $k$ and $2 \leq j \leq k-1$ be two positive integers, let $G_{j+1}$ be a complete graph on $j+1$ vertices, and let $c> 1$ be the smallest integer for which 
\begin{equation*}
k \leq R(\underbrace{G_{j+1},G_{j+1},\ldots,G_{j+1}}_{c \text{ }  times}).
\end{equation*}
Let $\epsilon = (\frac{j}{j+1})^c.$ Then all hollow symmetric nonnegative matrices of order at least $k$ and with off-diagonal entries from $(\epsilon,1]$ have at least $j$ nonpositive eigenvalues. 
\end{theorem}
As an immediate consequence of this theorem we have the following relevant result. 

\begin{corollary}
\label{lem1}
Let $\epsilon = (\frac{k-1}{k})^2$. If $\HB \in \mathbb{R}^{k \times k}$ is a hollow symmetric matrix with off-diagonal elements from $(\epsilon,1],$
then the largest eigenvalue of $\HB$ is positive and the other $k-1$ eigenvalues are nonpositive. 
\end{corollary}

\begin{proof}
In the special case of  $j = k-1$, we are trying find the smallest value of $c > 1$ for which the Ramsey number of $c$ copies of $G_k$ is greater than $k$, that is, the smallest value of $c$ for which
\begin{equation*}
k \leq R(\underbrace{G_{k},G_{k},\ldots,G_{k}}_{c \text{ }  times}).
\end{equation*}
For $c = 2$, the Ramsey number of $\{G_k, G_k\}$ is the minimum number of vertices needed for a complete graph such that any coloring results in an isomophic copy of $G_k$ whose edges are all monochromatic (one color). When $k=2$, any 2-coloring of $G_2$ contains a monochromatic copy of $G_2$, therefore $k = R(G_2,G_2)$. For $k>2$, there exists a 2-coloring of $G_k$ that does not contain a monochromatic isomorphic copy of $G_k$. For example, any 2-coloring that is not monochromatic will not contain a copy of $G_k$. Thus the number of vertices needed must be greater than $k$ and $k \leq R(G_k, G_k) \forall k$. This implies that $\HB$ will have at least $k-1$ nonpositive eigenvalues when the off-diagonal elements are chosen from the interval $(\epsilon,1]$ with $\epsilon =  (\frac{k}{k+1})^2$. Since the trace of a symmetric matrix is the sum of its eigenvalues, the eigenvalues of $\HB$ must sum to zero. This means that $k-1$ eigenvalues are nonpositive, but not all identically zero, implying that the largest eigenvalue must be positive. Thus $\HB$ has exactly one positive eigenvalue as desired.
\end{proof}

\subsection{Identifiability of the underlying correlation structure} 
\label{app2}

Theorem \ref{theorem:eval} and \ref{theorem:evec} state the conditions that allow us to determine the correlated components along with their correlation structure using the eigenvalue decomposition of $\CB$. One additional assumption in Theorem \ref{theorem:evec} is that the eigenvalues of $\CB$ greater than one are distinct, i.e., have algebraic multiplicity of one. In this section, we will briefly discuss why this assumption is needed. We will also mention the scenarios in which the correlation structure can still be completely determined using Theorem \ref{theorem:evec} even if the assumption is not true.  

Let $\lambda^{(i)}$ and $\lambda^{(j)}$ be the two eigenvalues of $\CB$ with $\lambda^{(i)} > 1$ and $\lambda^{(j)} > 1$. Let $\uB^{(i)}$ and $\uB^{(j)}$ be the eigenvectors associated with $\lambda^{(i)}$ and $\lambda^{(j)}$, respectively. Let $\uB = a\uB^{(i)} + b\uB^{(j)}$ be a vector formed by a linear combination of $\uB^{(i)}$ and $\uB^{(j)}$, and $a$ and $b$ are scalars. If $\lambda^{(i)} = \lambda^{(j)}$,  any linear combination of $\uB^{(i)}$ and $\uB^{(j)}$ is an eigenvector of $\lambda^{(i)}$ or $\lambda^{(j)}$. In this case, if the $i$th and $j$th group of components are correlated across different data sets, their correlation structure, i.e., across which data sets the components are correlated, cannot always be determined using Theorem \ref{theorem:evec}. For instance, if the $i$th components are correlated across all data sets except the $p$th data set, then according to Theorem \ref{theorem:evec}, the $p$th part of $\uB^{(i)}$, $\uB_p^{(i)} = \mathbf{0}$. Similarly, $\uB_q^{(j)} = \mathbf{0}$ if the $j$th components are correlated across all data sets except the $q$th data set. When $\lambda^{(i)} = \lambda^{(j)}$, then $a \uB^{(i)} + b\uB^{(j)}$ can also be an eigenvector of $\lambda^{(i)}$ or $\lambda^{(j)}$ for any $a,b$. Therefore, $\uB_p^{(i)}$ or $\uB_q^{(j)}$ are not necessarily equal to zero. 

However, if the $i$th and $j$th components are correlated across the same subset of data sets, even when $\lambda^{(i)} = \lambda^{(j)}$, their correlation structure can be determined using Theorem \ref{theorem:evec}. This is due to the fact that the zeros in $\uB = a\uB^{(i)} + b\uB^{(j)} $ are at the same positions as those of $\uB^{(i)}$ and $\uB^{(j)}$ for any $a,b$. 

To conclude, Theorem \ref{theorem:evec} can completely identify the correlation structure of the components when the eigenvalues associated with the components that are correlated across different subset of data sets are distinct. 
\end{appendix}

\bibliographystyle{ieeetran}
\bibliography{references_ex_summary} 

\end{document}